\documentclass[onecolumn,draftclsnofoot,12pt]{IEEEtran}
\usepackage{amsfonts}
\usepackage{times}
\usepackage{graphicx}
\usepackage{latexsym}
\usepackage{dsfont}
\usepackage{amssymb}
\usepackage{amsmath}
\usepackage{cite}
\usepackage{verbatim}
\usepackage{subfigure}

\newtheorem{assumption}{Assumption}

\newtheorem{lemma}{Lemma}

\newtheorem{proposition}{Proposition}

\newenvironment{proof}{ \textbf{Proof:} }{ \hfill $\Box$}

\newcommand{\figref}[1]{{Fig.}~\ref{#1}}

\newcommand{\fig}[1]{Fig.\ \ref{#1}}

\def\bb0{{\mathbb{0}}}

\def\ba{{\mathbf{a}}}
\def\bb{{\mathbf{b}}}

\def\bff{{\mathbf{f}}}
\def\bg{{\mathbf{g}}}
\def\bh{{\mathbf{h}}}

\def\bm{{\mathbf{m}}}

\def\bv{{\mathbf{v}}}

\def\bx{{\mathbf{x}}}

\def\b0{{\mathbf{0}}}

\def\bA{{\mathbf{A}}}

\def\bF{{\mathbf{F}}}

\def\bI{{\mathbf{I}}}

\def\bR{{\mathbf{R}}}

\def\bbE{{\mathbb{E}}}

\def\cA{\mathcal{A}}

\def\cF{\mathcal{F}}

\def\cN{\mathcal{N}}

\def\sf0{{\mathsf{0}}}

\usepackage{epstopdf}
\usepackage{enumerate}
\usepackage{algorithmicx}
\usepackage{algorithm}
\usepackage{amsmath}
\usepackage[noend]{algpseudocode}
\usepackage{float}
\usepackage{color}
\usepackage{makeidx}
\usepackage{bbm}
\usepackage{graphicx}
\usepackage{url}

\newcommand{\sref}[1]{{Section}~\ref{#1}}

\DeclareMathOperator*{\argmax}{arg\,max}

\begin{document}
\title{Deep Learning Coordinated Beamforming for Highly-Mobile  Millimeter Wave Systems}
\author{ Ahmed Alkhateeb, Sam Alex, Paul  Varkey, Ying Li, Qi Qu, and Djordje Tujkovic\\ \textit{Facebook}  \thanks{This work was done while the first author was with Facebook. Ahmed Alkhateeb is currently with Arizona State University (Email: alkhateeb@asu.edu).  Sam Alex, Paul  Varkey, Ying Li, Qi Qu, and Djordje Tujkovic are with Facebook, Inc., (Email: sampalex, paulvarkey, yingli, qqu, djordjet@fb.com).}}
\maketitle

\begin{abstract}

Supporting high mobility in millimeter wave (mmWave) systems enables a wide range of important
applications such as vehicular communications and wireless virtual/augmented reality. Realizing this
in practice, though, requires overcoming several challenges. First, the use of narrow beams and the
sensitivity of mmWave signals to blockage greatly impact the coverage and reliability of highly-mobile
links. Second, highly-mobile users in dense mmWave deployments need to frequently hand-off between
base stations (BSs), which is associated with critical control and latency overhead. Further, identifying
the optimal beamforming vectors in large antenna array mmWave systems requires considerable training
overhead, which significantly affects the efficiency of these mobile systems. In this paper, a novel
integrated machine learning and coordinated beamforming solution is developed to overcome these
challenges and enable highly-mobile mmWave applications. In the proposed solution, a number of
distributed yet coordinating BSs simultaneously serve a mobile user. This user ideally needs to transmit
only one uplink training pilot sequence that will be jointly received at the coordinating BSs using omni or
quasi-omni beam patterns. These received signals draw a defining signature not only for the user location,
but also for its interaction with the surrounding environment. The developed solution then leverages
a deep learning model that learns how to use these signatures to predict the beamforming vectors at
the BSs. This renders a comprehensive solution that supports highly-mobile mmWave applications with
reliable coverage, low latency, and negligible training overhead. Extensive simulation results, based on
accurate ray-tracing, show that the proposed deep-learning coordinated beamforming strategy approaches
the achievable rate of the genie-aided solution that knows the optimal beamforming vectors with no
training overhead, and attains higher rates compared to traditional mmWave beamforming techniques.
	
\end{abstract}

\newpage
%%%%%%%%%%%%%%%%%%%%%%%%%%%%%%%%%%%%%%%%%%%%%%%%%%%%%%%%%%%
\section{Introduction} \label{sec:Intro}

Millimeter wave (mmWave) communication attracted considerable interest in the last few years, thanks to the high data rates enabled by its large available bandwidth \cite{HeathJr2016,Andrews2017,Rappaport2013a}. This makes mmWave a key technology for next-generation wireless systems \cite{Boccardi2014,Roh2014,Hur2016,11ad}. Most of the prior research has focused on developing beamforming strategies \cite{ElAyach2014,Alkhateeb2014d,Molisch2017}, evaluating the performance \cite{Bai2014,Singh2015,Zhu2014}, or studying the practical feasibility of mmWave communication at fixed or low-mobility wireless systems \cite{Cudak2014,Ghosh2014,Hong2014}. 
But can mmWave also support highly-mobile yet data-hungry applications, such as  vehicular communications or wireless augmented/virtual reality (AR/VR)? Enabling these applications faces several critical challenges:  
(i) the sensitivity of mmWave signal propagation to blockages and the large signal-to-noise ratio (SNR) differences between line-of-sight (LOS) and non-LOS links severely impact the \textbf{reliability} of the mobile systems, 
(ii) with mobility, and in dense deployments, the user needs to frequently hand over from one base station (BS) to another, which imposes control overhead and introduces a \textbf{latency} problem, 
and (iii) with large arrays, adjusting the beamforming vectors requires \textbf{large training overhead}, which imposes a fundamental limit on supporting mobile users. 
In this paper, we develop a novel solution based on coordinated beamforming, and leveraging tools from machine learning, to jointly address all these challenges and enable highly-mobile mmWave systems.

\subsection{Prior Work}
Coordinating the transmission between multiple BSs to simultaneously serve the same user is one main solution for enhancing the coverage and overcoming the frequent handover problems in dense mmWave systems \cite{Jr.2017,Maamari2016,Gupta2018}. 
In \cite{Jr.2017}, extensive measurements for 73 GHz coordinated multi-point transmission were done at an urban open square scenario in downtown Brooklyn. The measurements showed that serving a user simultaneously by a number of BSs achieves significant coverage improvement.
Analyzing the network coverage of coordinated mmWave beamforming was also addressed in prior work \cite{Maamari2016,Gupta2018}, mainly using tools from stochastic geometry. 
In \cite{Maamari2016}, the performance of heterogeneous mmWave cellular networks was analyzed to show that a considerable coverage gain can be achieved using base station cooperation, where the user is simultaneously served by multiple BSs.  
In \cite{Gupta2018}, a setup where the user is only connected to LOS BSs was considered and the probability of having at least one LOS BS was analyzed. The results showed that the density of coordinating BSs should scale with the square of the blockage density to maintain the same LOS connectivity. 
While \cite{Jr.2017,Maamari2016,Gupta2018} proved the significant coverage gain of BS coordination, they did not investigate how to construct these coordinated beamforming vectors, which is normally associated with high coordination overhead. This paper, therefore, targets developing low-complexity mmWave coordination strategies that harness the coordination coverage and latency gains but with low coordination overhead.  

The other major challenge with highly-mobile mmWave systems is the huge training overhead associated with adjusting large array beamforming vectors. Developing beamforming/channel estimation solutions to reduce this training overhead has attracted considerable research interest in the last few years \cite{Wang2009,Hur2013,Noh2017,Donno2017,Alkhateeb2014,Ramasamy2012,Rasekh2017,Schniter2014,Han2016,Abdelreheem2016,Alexandropoulos2017,Garcia2016,Choi2016,Va2017a}. This prior research has mainly focused on three directions: (i) beam training \cite{Wang2009,Hur2013,Noh2017,Donno2017}, (ii) compressive channel estimation \cite{Alkhateeb2014,Ramasamy2012,Rasekh2017,Schniter2014,Han2016}, and (iii) location aided beamforming \cite{Abdelreheem2016,Alexandropoulos2017,Garcia2016,Choi2016,Va2017a}. In beam training, the candidate beams at the transmitter and receiver are directly trained using exhaustive or adaptive search to select the ones that optimize the metric of interest, e.g., SNR. Beam training, though, requires large overhead to train all the possible beams and is mainly suitable for single-user and single stream transmissions \cite{Wang2009,Hur2013,Noh2017,Donno2017}. In order to enable spatial multiplexing at mmWave systems, \cite{Alkhateeb2014,Ramasamy2012,Rasekh2017,Schniter2014,Han2016} proposed to leverage the sparsity of mmWave channels and formulated the mmWave channel estimation problem as a sparse reconstruction problem. Compressive sensing tools were then used to efficiently estimate the parameters (angles of arrival/departure, path gains, etc.) of the sparse channel. While compressive channel estimation techniques can generally reduce the training overhead compared to exhaustive search solutions, they still require relatively large training overhead that scales with the number of antennas. Further, compressive channel estimation techniques normally make hard assumptions on the exact sparsity of the channel and the quantization of the angles of arrival/departure, which leaves their practical feasibility uncertain. 

To further reduce the training overhead, and given the directivity nature of mmWave beamforming, out-of-band information such as the locations of the transmitter and receiver can be leveraged to reduce the beamforming training overhead \cite{Abdelreheem2016,Alexandropoulos2017,Garcia2016,Choi2016,Va2017a}. In \cite{Abdelreheem2016}, the transmitter/receiver location information was exploited to guide the sensing matrix design used in the compressive estimation of the channel. Position information was also leveraged in \cite{Alexandropoulos2017,Garcia2016} to build the beamforming vectors in LOS mmWave backhaul and vehicular systems. In \cite{Choi2016,Va2017a}, the BSs serving vehicular systems build a database relating the vehicle position and the beam training result. This database is then leveraged to reduce the training overhead with the knowledge of the vehicle location.  While the solutions in \cite{Abdelreheem2016,Alexandropoulos2017,Garcia2016,Choi2016,Va2017a} showed that the position information can reduce the training overhead, relying only on the location information to design the beamforming vectors has several limitations. First,  position-acquisition sensors, such as GPS, have limited accuracy, normally in the order of meters, which may not work efficiently with narrow-beam systems. Second, GPS sensors do not work well inside buildings, which makes these solutions not capable of supporting indoor applications. Further, the beamforming vectors are not merely a function of the transmitter/receiver location but also of the environment geometry, blockages, etc. This makes location-based beamforming solutions mainly suitable for LOS environment, as the same location in NLOS environment may correspond to different beamforming vectors depending, for example, on the position of the obstacles.

\subsection{Contribution}

In this paper, we propose a novel integrated communication and machine learning solution for highly-mobile mmWave applications. Our proposed solution considers a coordinated beamforming system where a set of BSs simultaneously serve one mobile user. For this system, {a deep learning model learns how to predict the BSs beamforming vectors directly from  the signals received at the distributed BSs using only \textit{omni} or quasi-omni beam patterns}. This is motivated by the intuition that the signals jointly received at the  distributed BSs draw a defining multi-path signature not only of the user location, but also of its surrounding environment. This proposed solution has multiple gains.
First, making beamforming prediction based on the uplink received signals, and not on position information, enables the developed strategy to support both LOS and NLOS scenarios and waves the requirement for special position-acquisition sensors. Second, the prediction of the optimal beams requires only \textit{omni} received pilots, which can be captured with negligible training overhead. Further, the deep learning model in the proposed system operation does not require any training before deployment, as it learns and adapts to any environment. Finally, since the proposed deep learning model is integrated with the coordinated beamforming system, it inherits the coverage and reliability gains of coordination. 
More specifically, this paper contributions can be summarized as follows. 

\begin{itemize}
	\item We propose a low-complexity coordinated beamforming system in which a number of BSs adopting RF beamforming, linked to a central cloud processor applying baseband processing, simultaneously serve a mobile user. 
	%The proposed system is motivated by the reliable-coverage and low hand-off latency promised by coordinated beamforming in highly-mobile mmWave applications. 
	For this system, we formulate the training and design problem of the central baseband and BSs RF beamforming vectors to maximize the system \textit{effective} achievable rate. The effective rate is a metric that accounts for the trade-off between the beamforming training overhead and achievable rate with the designed beamforming vectors, which makes it a suitable metric for highly-mobile mmWave systems. 
	
	\item We develop a baseline coordinated beamforming strategy for the adopted system, which depends on uplink training in designing the RF and baseband beamforming vectors. With this baseline solution, the BSs first select their RF beamforming vectors from a predefined codebook. Then, a central processor designs its baseband beamforming to ensure coherent combining at the user. We prove that in some special yet important cases, the baseline beamforming strategy obtains optimal achievable rates. This solution, though, requires high training overhead, which motivates the integration with machine learning models. 
	
	\item We propose a novel integrated deep learning and coordinated beamforming solution, and develop its system operation and machine learning modeling. The key idea of the proposed solution is to leverage the signals received at the coordinating BSs with only omni or quasi-omni patterns, i.e., with negligible training overhead, to predict their RF beamforming vectors.  Further, the developed solution enables harvesting the wide-coverage and low-latency coordinated beamforming gains with low coordination overhead, rendering it a promising enabling solution for highly-mobile mmWave applications. 
\end{itemize}

\noindent Extensive simulations were performed to evaluate the performance of the developed solution and the impact of the key system and machine learning parameters. 
At both LOS and NLOS scenarios, the results show that the effective achievable rate of the developed solution approaches that of the genie-aided coordinated beamforming which knows the optimal beamforming vectors with no training overhead. 
Compared to the baseline solution, deep-learning coordinated beamforming achieves a noticeable gain, especially when users are moving with high speed and when the BSs deploy large antenna arrays. 
The results also confirm the ability of the proposed deep learning based beamforming to learn and adapt to time-varying environment, which is important for the system robustness.
Further, the results show that learning coordinated beamforming may not require phase synchronization among the coordinating BSs, which is especially important for practical implementations. 
All that highlights the capability of the proposed deep-learning solution in efficiently supporting highly-mobile applications in large-array mmWave systems.

% Notation
\textbf{Notation}: We use the following notation throughout this paper: $\bA$ is a matrix, $\ba$ is a vector, $a$ is a scalar, and $\cA$ is a set. $|\bA|$ is the determinant of $\bA$, whereas $\bA^T$, $\bA^H$, $\bA^*$ are its transpose, Hermitian (conjugate transpose), and conjugate respectively. $\mathrm{diag}(\ba)$ is a diagonal matrix with the entries of $\ba$ on its diagonal, and $\mathrm{blkdiag} \left(\bA_1, ..., \bA_N\right)$ is a block diagonal matrix with the matrices $\bA_1, ..., \bA_N$ on the diagonal. $\bI$ is the identity matrix and  $\cN(\bm,\bR)$ is a complex Gaussian random vector with mean $\bm$ and covariance $\bR$.

%%%%%%%%%%%%%%%%%%%%%%%%%%%%%%%%%%%%%%%%%%%%%%%%%%%%%%%%%%%%%%%%
\section{System and Channel Models} \label{sec:Model}
%%%%%%%%%%%%%%%%%%%%%%%%%%%%%%%%%%%%%%%%%%%%%%%%%%%%%%%%%%%%%%%
In this section, we describe the adopted frequency-selective coordinated mmWave system and channel models. The key assumptions made for each model are also
highlighted.

\subsection{System Model} \label{sec:SysModel}
\begin{figure}[t]
	\centerline{
		\includegraphics[scale=.75]{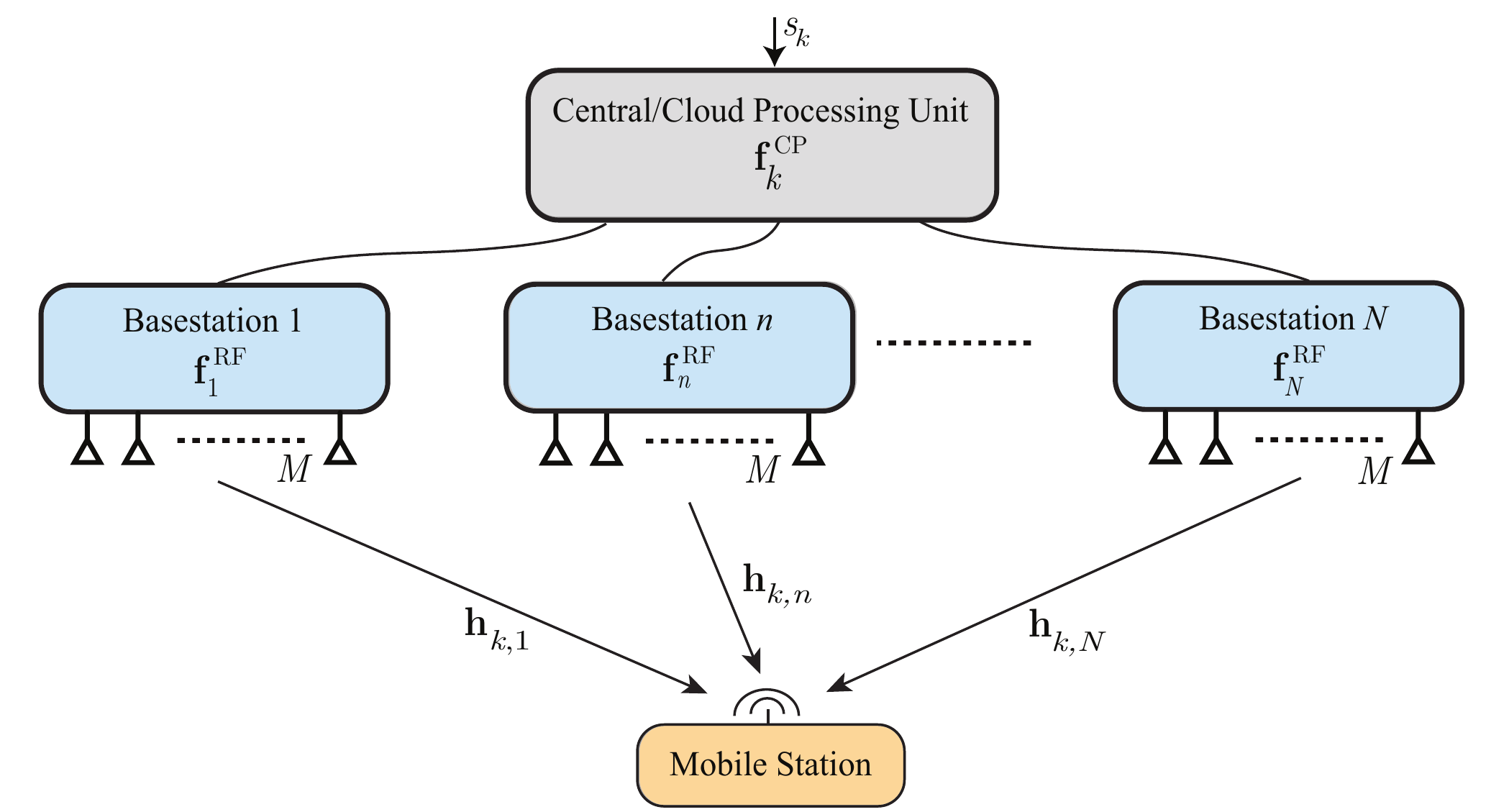}
	}
	\caption{A block diagram of the proposed  mmWave coordinated beamforming system. The transmitted signal $s_k$ at every subcarrier $k, k=1, .., K,$ is first precoded at the central/cloud processing unit using $\bff_k^\mathrm{CP}$, and then transmitted jointly from the $N$ terminals/BSs employing the RF beamforming vectors $\bff^\mathrm{RF}_{n}, n=1, ..., N$.}
	\label{fig:Sys_Model}
\end{figure}

Consider the mmWave communication system in \figref{fig:Sys_Model}, where $N$ base stations (BSs) or access points (APs) are simultaneously serving one mobile station (MS). Each BS is equipped with $M$ antennas and all the BSs are connected to a centralized/cloud processing unit. For simplicity, we assume that every BS has only one RF chain and is applying analog-only beamforming using networks of phase shifters \cite{HeathJr2016}. Extensions to more sophisticated mmWave precoding architectures at the BSs such as hybrid precoding \cite{ElAyach2014,Alkhateeb2014d} are also interesting for future research. In this paper, we assume that the mobile user has a single antenna. The developed algorithms and solutions, though, can be extended to multi-antenna users.

\textbf{In the downlink transmission},  the data symbol $s_k \in \mathbb{C}$ at subcarrier $k, k=1, ..., K,$ is first precoded using the $N \times 1$ digital precoder $\bff^\mathrm{CP}_k = \left[f^\mathrm{CP}_{k,1}, ..., f^\mathrm{CP}_{k,N}\right]^T \in \mathbb{C}^{N \times 1}$ at the central/cloud processing unit. The resulting symbols are transformed to the time domain using $N$ $K$-point IFFTs. A cyclic prefix of length $D$ is then added to the symbol blocks before sending them to the BSs using error-negligible and delay-negligible wired channels, e.g., optical fiber cables. Every BS $n$ applies a time-domain analog beamforming $\bff^\mathrm{RF}_n \in \mathbb{C}^{M \times 1}$ and transmits the resulting signal. The discrete-time transmitted complex baseband signal from the $n$th BS at the $k$th subcarrier can then be written as
\begin{equation} \label{eq:Tx_Signal}
\bx_{k,n}= \bff^\mathrm{RF}_n f^\mathrm{CP}_{k,n} s_k, 
\end{equation} 
where the transmitted signal on the $k$-th subcarrier  $s_k$ is normalized such that $\bbE [ s_k s_k^H]=\frac{P}{K}$, with $P$ the average total transmit power. Since the RF beamforming is assumed to be implemented using networks of quantized phase shifters, the entries of $ \bff^\mathrm{RF}_n$ are modeled as $\left[\bff^\mathrm{RF}_n\right]_m=\frac{1}{\sqrt{M}}e^{j \overline{\phi}_{n,m}}$, where $\overline{\phi}_{n,m}$ is a quantized angle. Adopting a per-subcarrier transmit power constraint and defining $\bF^\mathrm{RF}=\mathrm{blkdiag} \left( \bff^\mathrm{RF}_1, ..., \bff^\mathrm{RF}_N\right) \in \mathbb{C}^{N M \times N}$ , the cloud baseband precoder and the BSs RF beamformers satisfy
\begin{equation}
\left\| \bF^\mathrm{RF} \bff^\mathrm{CP}_k \right\|^2 = 1, \ \ \ k=1, 2, ..., K.
\end{equation}

At the user, assuming perfect frequency and carrier offset synchronization, the received signal is transformed to the frequency domain using a K-point FFT. Denoting the $M \times 1$ channel vector between the user and the $n$th BS at the $k$th subcarrier as $\bh_{k,n} \in \mathbb{C}^{M \times 1}$, the received signal at subcarrier $k$ after processing can be expressed as 
\begin{equation}
y_k= \sum_{n=1}^N \bh_{k,n}^T \bx_{k,n} + v_k, 
\end{equation}
where  $v_k \sim \mathcal{N}_\mathbb{C} \left(0,\sigma^2\right)$ is the receive noise at subcarrier $k$. 

\subsection{Channel Model} \label{sec:ChModel}

We adopt a geometric wideband mmWave channel model \cite{Rappaport2013a,Akdeniz2014,11ad,Samimi2014} with $L$ clusters. Each cluster $\ell, \ell=1, ..., L$ is assumed to contribute with one ray that has a time delay $\tau_\ell \in \mathbb{R}$, and azimuth/elevation angles of arrival (AoA) $\theta_\ell, \phi_\ell$. Further, let $\rho_{n}$ denote the path-loss between the user and the $n$-th BS, and $p_\mathrm{rc}(\tau)$ represents a pulse shaping function for $T_S$-spaced signaling evaluated at $\tau$ seconds \cite{Schniter2014}. With this model, the delay-d channel vector between the user and the $n$th BS, $\mathsf{\boldsymbol{h}}_{d,n}$, follows
\begin{equation} \label{eq:d-channel}
\mathsf{\boldsymbol{h}}_{d,n} = \sqrt{\frac{M}{\rho_{n}}} \sum_{\ell=1}^L \alpha_\ell \hspace{1pt}  p(d T_\mathrm{S} - \tau_\ell) \hspace{1pt} \ba_n\left(\theta_\ell, \phi_\ell\right),
\end{equation}
where  $\ba_n\left(\theta_\ell,\phi_\ell\right)$ is the array response vector of the $n$th BS at the AoA $\theta_\ell, \phi_\ell$.  
Given the delay-d channel in \eqref{eq:d-channel}, the frequency domain channel vector at subcarrier $k$, $\bh_{k,n}$, can be written as 
\begin{equation}
\bh_{k,n}=\sum_{d=0}^{D-1}  \mathsf{\boldsymbol{h}}_{d,n}  e^{-j \frac{2 \pi k}{K} d}.
\end{equation} 

\noindent Considering a block-fading channel model, $\left\{\bh_{k,n}\right\}_{k=1}^K$ are assumed to stay constant over the channel coherence time, denoted $T_\mathrm{C}$, which depends on the user mobility and the channel multi-path components  \cite{Va2015} . In the next section, we will develop the problem formulation and discuss this channel coherence time in more detail. 
 
%%%%%%%%%%%%%%%%%%%%%%%%%%%%%%%%%%%%%%%%%%%%%%%%%%%%%%%%%%
\section{Problem Formulation} \label{sec:Problem}
%%%%%%%%%%%%%%%%%%%%%%%%%%%%%%%%%%%%%%%%%%%%%%%%%%%%%%%%%%
The main goal of the proposed coordinated mmWave beamforming system is to enable wireless applications with high mobility and high data rate requirements, and with strict constraints on the coverage, reliability, and latency. Thanks to simultaneously serving the user from multiple BSs, the coordinated beamforming system in \sref{sec:Model} provides transmission diversity and robustness against blockage, which directly enhances the system coverage, reliability, and latency. The main challenge, however, with this system is achieving the high data rate requirements, as the time overhead of training and designing the cloud baseband and terminals RF beamforming vectors can be very large, especially for highly-mobile users. With this motivation, this paper focuses on developing efficient channel training and beamforming design strategies that maximize the system \textit{effective} achievable rate, and enable highly-mobile mmWave applications. Next, we formulate the effective achievable rate optimization problem.

\textbf{Achievable Rate:} Given the system and channel models in \sref{sec:Model}, and employing the cloud and RF beamformers $\left\{\bff^\mathrm{CP}_k\right\}_{k=1}^K$, $\bF^\mathrm{RF}$, the user achievable rate is expressed as 
\begin{equation} \label{eq:Ach_R}
R=\frac{1}{K}\sum_{k=1}^K  \log_2\left(1+\mathsf{SNR} \left| \sum_{n=1}^N \bh_{k,n}^T  \bff^\mathrm{RF}_n f^\mathrm{CP}_{k,n}      \right|^2 \right),
\end{equation}
where  $\mathsf{SNR}=\frac{P}{K \sigma^2}$ denotes the signal-to-noise ratio.

 Due to the constraints on the RF hardware, such as the availability of only quantized angles, $\overline{\phi}_{m,n}$, for the RF phase shifters, the BSs RF beamforming vectors $\bff^\mathrm{RF}_{n}, n=1, ..., N$, can take only certain values \cite{ElAyach2014,Wang2009,Wang2015,Alkhateeb2014}. Therefore, we assume that the RF beamforming vectors are selected from finite-size codebooks, which we formally state in the following assumption.
\begin{assumption}
	The BSs RF beamforming vectors are subject to the quantized codebook constraint, $\bff^\mathrm{RF}_{n} \in \boldsymbol{\cF}_\mathrm{RF}, \forall n$, where the cardinality of $\boldsymbol{\cF}_\mathrm{RF}$ is $|\boldsymbol{\cF}_\mathrm{RF}|=N_\mathrm{tr}$.
\end{assumption}

\noindent The optimal cloud baseband and terminals RF beamforming vectors that maximize the system achievable rate can then be found by solving 
\begin{align} \label{eq:Opt1}
\{ \overset{\star}{\bff^{\mathrm{CP}}_k}\}_{k=1}^{K}, \{\overset{\star}{\bff^\mathrm{RF}_{n}}\}_{n=1}^N = & \argmax \sum_{k=1}^K \log_2\left(1+\mathsf{SNR} \left|\sum_{n=1}^N \bh_{k,n}^T \bff^\mathrm{RF}_n f^\mathrm{CP}_{k,n}     \right|^2\right), \\
& \text{s.t.} \hspace{30pt} \bff^\mathrm{RF}_{n} \in \boldsymbol{\cF}_\mathrm{RF},  \hspace{30pt} \forall n, \\
&   \hspace{40pt} \left\| \bF^\mathrm{RF}  \bff^\mathrm{CP}_{k} \right\|^2 = 1, \hspace{30pt} \forall k, \label{eq:Opt_2L3}
\end{align}
which is addressed in the next lemma.

\begin{lemma} 
\label{lem:1} 
 For a given channel $\overline{\bh}_{k}=\left[\bh_{k,1}^T, ..., \bh_{k,N}^T\right]^T, \forall k$, the optimal cloud baseband precoder and terminal RF beamformers that solve \eqref{eq:Opt1}-\eqref{eq:Opt_2L3} are 
	\begin{equation} \label{eq:Sol1}
	\overset{\star}{\bff^{\mathrm{CP}}_k}=\frac{{\left(\overline{\bh}_{k}^T \bF^\mathrm{RF}\right)^H }}{\left\| \overline{\bh}_{k}^T \bF^\mathrm{RF}  \right\|}, \hspace{30pt}  \forall k,
	\end{equation}
	and 
	\begin{equation} \label{eq:Sol2}
	\{\overset{\star}{\bff^\mathrm{RF}_{n}}\}_{n=1}^N =\argmax_{\bff^\mathrm{RF}_n \in \boldsymbol{\cF_\mathrm{RF}}, \forall n}    \sum_{k=1}^K \log_2\left( 1+ \mathsf{SNR} \sum_{n=1}^N \left|\bh_{k,n}^T  \bff^\mathrm{RF}_n \right|^2   \right),
	\end{equation}
	which yield the optimal achievable rate $R^\star$.
\end{lemma} 
\begin{proof}
	The proof is straightforward, and follows from the maximum ratio transmit solution by noting that the power constraint $\left\| \bF^\mathrm{RF}  \bff^\mathrm{CP}_{k} \right\|^2 = 1$ can be reduced to $\left\| \bff^\mathrm{CP}_{k} \right\|^2 = 1$, given the block diagonal structure of the RF precoding matrix $\bF^\mathrm{RF}$.
\end{proof}

\textbf{Effective Achievable Rate:}  The optimal achievable rate $R^\star$, given by Lemma \ref{lem:1}, assumes perfect channel knowledge at the cloud processing unit and RF terminals. Obtaining this channel knowledge, however, is very challenging and requires large training overhead in mmWave systems with RF architectures. This is mainly due to (i) the large number of antennas at the BSs,  and (ii) the RF filtering of the channel seen at the baseband \cite{Alkhateeb2014d}. To accurately evaluate the actual rate experienced by the mobile user, it is important to incorporate the impact of this time overhead required for the channel training and beamforming design. For that, we adopt the \textit{effective} achievable rate metric, which we define shortly. 

The formulation of the effective achievable rate requires understanding how often the beamforming vectors need to be redesigned as the user moves. This can be captured by one of two metrics: (i) the channel coherence time $T_\mathrm{C}$, which is the time over which the multi-path channel remains almost constant, and (ii) the channel \textit{beam} coherence time $T_\mathrm{B}$, which is a recent concept introduced for mmWave systems to represent the average time over which the beams stay aligned \cite{Va2015}. While the channel coherence time is normally shorter than the beam coherence time, It was shown in \cite{Va2015} that updating the beams every beam coherence time incurs negligible receive power loss compared to updating them every channel coherence time. Adopting this model, we make the following assumption on the system operation. 
\begin{assumption} \label{assum2}
	The cloud baseband and terminal RF beamforming vectors are assumed to be retrained and redesigned every beam coherence time, $T_\mathrm{B}$, such that the first $T_\mathrm{tr}$ time of every beam coherence time  is allocated for the channel training and beamforming design, and the rest of it is used for the data transmission using the designed beamforming vectors. 
\end{assumption}

Now, we define the effective achievable rate, $R_\mathrm{eff}$,  as the achievable rate using certain precoders, $\left\{\bff^\mathrm{CP}_k\right\}_{k=1}^K, \bF^\mathrm{RF}$, times the percentage of time these precoders are used for data transmission, i.e.,
\begin{equation} \label{eq:effR}
R_\mathrm{eff}= \left(1-\frac{T_\mathrm{tr}}{T_\mathrm{B}}\right) \frac{1}{K}\sum_{k=1}^K  \log_2\left(1+\mathsf{SNR} \left| \sum_{n=1}^N \bh_{k,n}^T  \bff^\mathrm{RF}_n f^\mathrm{CP}_{k,n}      \right|^2 \right).
\end{equation}
The effective achievable rate in \eqref{eq:effR} captures the impact of user mobility on the actually experienced data rate. For example, with higher mobility, the beam coherence time decreases, which results in lower data rate for the same beamforming vectors and beam training overhead.  \textbf{The objective of this paper} is then to develop efficient channel training and beamforming design strategies that maximize the system effective achievable rate. If $\Pi\left(T_\mathrm{tr}, \left\{\bff^\mathrm{CP}_k\right\}_{k=1}^K, \bF^\mathrm{RF}\right)$ represents a certain channel training/beamforming design strategy that requires training overhead $T_\mathrm{tr}$ to design the cloud and RF beamforming vectors $\left\{\bff^\mathrm{CP}_k\right\}_{k=1}^K, \bF^\mathrm{RF}$, the final problem formulation can then be written as
\begin{align} \label{eq:Opt2}
\Pi^\star\left(T_\mathrm{tr}, \left\{\bff^\mathrm{CP}_k\right\}_{k=1}^K, \bF^\mathrm{RF}\right)  = &  \argmax \left(1-\frac{T_\mathrm{tr}}{T_\mathrm{B}}\right) \sum_{k=1}^K \log_2\left(1+\mathsf{SNR} \left|\sum_{n=1}^N \bh_{k,n}^T \bff^\mathrm{RF}_n f^\mathrm{CP}_{k,n}      \right|^2\right), \\
&  \text{s.t.} \hspace{30pt} \bff^\mathrm{RF}_n \in \cF_\mathrm{RF}  \hspace{30pt} \forall n, \\
&    \hspace{40pt} \left\| \bff^\mathrm{CP}_k\right\|^2 = 1 \hspace{30pt} \forall k.  \label{eq:Opt_3_3}
\end{align}

Solving the problem in \eqref{eq:Opt2}-\eqref{eq:Opt_3_3} means developing solutions that require very low channel training overhead to realize beamforming vectors that maximize the system achievable rate, $R$. It is worth noting also that $R^\star$ represents an ultimate upper bound for the effective achievable rate $R_\mathrm{eff}$ with $T_\mathrm{tr}=0$ and $R=R^\star$. 

In the literature, two main directions to address this mmWave channel estimation/beamforming design problem are compressed sensing and beam training. In compressed sensing, the sparsity of  mmWave channels is leveraged and random beams are employed to estimate the multi-path channel parameters, such as the angles or arrival and path gains  \cite{Alkhateeb2014,Alkhateeb2015,Ramasamy2012,Schniter2014,Rasekh2017}. The estimated channel can then be used to construct the beamforming vectors. The other approach is to directly train the RF beamforming vectors through exhaustive or hierarchical search to find the best beams \cite{Wang2009,Hur2013,11ad}.  Each of the two directions has its own advantages and limitations. Both of them, though, require large training overhead which makes them inefficient in handling highly-mobile mmWave applications.  In this paper, we show that integrating machine learning tools with typical mmWave beam training solutions can yield efficient channel training/beamforming design strategies that have very low training overhead and near-optimal achievable rates, which enables highly-mobile mmWave systems.

In the next sections, we present a baseline coordinated mmWave beamforming solution based on conventional beam training techniques. Then, we show in  \sref{sec:DL_Framework} how machine learning models can be integrated with the proposed baseline solution, leading to novel techniques with near-optimal effective achievable rates for mmWave systems.

%%%%%%%%%%%%%%%%%%%%%%%%%%%%%%%%%%%%%%%%%%%%%%%%%%%%%%%%%%%%%%%
\section{Baseline Coordinated Beamforming} \label{sec:BL_Comm}
%%%%%%%%%%%%%%%%%%%%%%%%%%%%%%%%%%%%%%%%%%%%%%%%%%%%%%%%%%%%%%%

In this section, we present a baseline solution for the channel training/beamforming design problem in \eqref{eq:Opt2}-\eqref{eq:Opt_3_3} based on conventional communication system tools. The proposed solution has low beamforming design complexity and enables the integration with the machine learning model in \sref{sec:DL_Framework}. In the following subsections, we present the baseline solution and evaluate its achievable rate performance and mobility support. 

\subsection{Proposed Solution} \label{subsec:BL}

As shown in \sref{sec:Problem}, for a given set of RF beamforming vectors $\left\{\bff^\mathrm{RF}_n\right\}_{n=1}^N$, the cloud baseband beamformers can be written optimally as a function of the effective channel $\overline{\bh}_{k}^T \bF^\mathrm{RF}$. This implies that the cloud baseband and terminal RF beamforming design problem is separable and can be solved in two stages for the RF and baseband beamformers. To find the optimal RF beamforming vectors, though, an exhaustive search over all possible BSs beamforming combinations is needed, as indicated in \eqref{eq:Sol2}. This yields high computational complexity, especially for large antenna systems with large codebook sizes. For the sake of low-complexity solution, we propose the following system operation. 

%%%-------------------------------------------------------%%%
\textbf{Uplink Simultaneous Beam Training:} 
%%%-------------------------------------------------------%%%
In this stage, the user transmits $N_\mathrm{tr}=\left|\boldsymbol{\cF}_\mathrm{RF}\right|$ repeated pilot sequences of the form $\left\{s_k^{\mathrm{pilot}}\right\}_{k=1}^K$ to the BSs. During this training time, every BS switches between its $N_\mathrm{tr}$ RF beamforming vectors such that it combines every received pilot sequence with a different RF beamforming vector. Let $\bg_p, p=1, ..., N_\mathrm{tr}$ denotes the $p$-th beamforming codeword in $\boldsymbol{\cF}$, then the combined received signal at the $n$-th BS for the $p$-th training sequence can be expressed as
\begin{equation}
r_{k,n}^{(p)}= \bg_{p}^T \bh_{k,n} s^\mathrm{pilot}_k+ \bg_p^T \bv_{k,n}, \hspace{50pt} k=1, 2, ..., K,
\end{equation}
where $\bv_{k,n} \sim \mathcal{N}_\mathbb{C}\left(\boldsymbol{0}, \sigma^2 \bI\right)$ is the receive noise vector at the $n$-th BS and $k$-th subcarrier. 

The combined signals for all the beamforming codewords are then fed back from all the BSs/terminals to the cloud processor, which calculates the received power using every RF beamforming vector and selects the BSs downlink RF beamforming vectors separately for every BS, according to 
\begin{align} 
\bff^\mathrm{BL}_n =\argmax_{\bg_p \in \boldsymbol{\cF}_\mathrm{RF} } \sum_{k=1}^K  \log_2\left(1+\mathsf{SNR} \left| \bh_{k,n}^T  \bg_{p}     \right|^2 \right). \label{eq:BL_RF}
\end{align}
Note that selecting the RF beamforming vectors disjointly for the different BSs avoids the combinatorial optimization complexity of the exhaustive search and enables the integration with the machine learning model, as will be discussed in \sref{sec:DL_Framework}. Further, this disjoint optimization can be shown to yield optimal achievable rate in some important special cases for mmWave systems, which will be discussed in the next subsection. Once the RF beamforming vectors are selected, the cloud baseband beamforming vectors are constructed according to \eqref{eq:Sol1}. 

%%%-------------------------------------------------------%%%
\textbf{Downlink Coordinated Data Transmission:}
%%%-------------------------------------------------------%%%
The designed cloud and RF beamforming vectors are employed for the downlink data transmission to achieve the coverage, reliability, and latency gains of the coordinated beamforming transmission. With the proposed baseline solution for the channel training/beamforming design, and denoting the beam training pilot sequence time as $T_\mathrm{p}$, the effective achievable rate, $R_\mathrm{eff}^\mathrm{BL}$, can be characterized as 
\begin{equation} \label{eq:EfFR_BL}
R_\mathrm{eff}^\mathrm{BL}=\left(1-\frac{N_\mathrm{tr} T_\mathrm{p}}{T_\mathrm{B}}\right) \frac{1}{K} \sum_{k=1}^K \log_2\left(1+  \mathsf{SNR} \sum_{n=1}^N \left|\bh_{k,n}^T \bff^\mathrm{BL}_n \right|^2  \right),
\end{equation}
where  the RF beamforming vectors $\bff^\mathrm{BL}_n, n=1, ..., N$, are given by \eqref{eq:BL_RF}. 

%%%-------------------------------------------------------%%%
\subsection{Performance Analysis and Mobility Support}
%%%-------------------------------------------------------%%%
In this subsection, we evaluate the achievable rate performance of the proposed solution and discuss its 
mobility support.

%%%-------------------------------------------------------%%%
\textbf{Achievable Rate:}
%%%-------------------------------------------------------%%%
Despite its low complexity and the disjoint RF beamforming design, the achievable rate of the baseline coordinated beamforming solution converges to the upper bound $R^\star$ in important special cases for mmWave systems, namely in the single-path channels and large antenna regimes, which is captured by the following proposition. 
\begin{proposition} \label{prop:BL_R}
Consider the system and channel models in \sref{sec:Model}, with a pulse shaping function $p(t) =\delta(t)$, then the achievable rate of the baseline coordinated beamforming solution satisfies 
\begin{equation}
R^\mathrm{BL}= \frac{1}{K} \sum_{k=1}^K \log_2\left(1+  \mathsf{SNR}  \sum_{n=1}^N \left|\bh_{k,n}^T \bff^\mathrm{BL}_n \right|^2  \right) \ = R^\star, \ \ \ \text{for} \ L=1,
\end{equation}
and when a beamsteering codebook $\boldsymbol{\cF}_\mathrm{RF}$ is adopted, with beamforming codewords $\bg_p=\ba(\overline{\theta}_p, \overline{\phi}_p)$ for some quantized angles $\overline{\theta_p}, \overline{\phi}_p$, the achievable rate of the baseline solution follows
\begin{equation}
\lim_{M \rightarrow \inf} R^\mathrm{BL} = R^\star \ \ \  \text{almost surely}. 
\end{equation}
\end{proposition}
\begin{proof}
The proof is simple and is omitted due to space limitation. 
\end{proof}

Proposition \ref{prop:BL_R} shows that, for some important special cases, the disjoint RF beamforming design across BSs achieves the same data rate of the upper bound $R^\star$ which requires combinatorial optimization complexity. 

%%%-------------------------------------------------------%%%
\textbf{Effective Achievable Rate and Mobility Support:}
%%%-------------------------------------------------------%%%
The effective achievable rate depends on (i) the time overhead in training the channel and designing the beamforming vectors, and (ii) the achievable rate using the constructed beamforming vectors. While the baseline solution can achieve optimal rate in some special yet important mmWave-relevant cases, the main drawback of this solution is the requirement of large training overhead, as it exhaustively searches over all the $N_\mathrm{tr}$ codebook beamforming vectors. This makes it very inefficient in supporting wireless applications with high throughput and mobility requirements. For example, consider a system model with BSs employing $32 \times 8$ uniform planar antenna arrays, and adopting an oversampled beamsteering RF codebook of size $N_\mathrm{tr}=1024$. If the pilot sequence training time is $T_\mathrm{p}=10$ us, this means that the training over head will consume $\sim 45\%$ of the channel beam coherence time for a vehicle moving with speed $v=30$ mph, whose beam coherence time is around $23$ ms \cite{Va2015}.  In the next section, we show how machine learning can be integrated with this baseline solution to dramatically reduce this training overhead and enable highly-mobile mmWave applications. 

%%%%%%%%%%%%%%%%%%%%%%%%%%%%%%%%%%%%%%%%%%%%%%%%%%%%%%%%%%%%%%%%%%%%%%%%%%%%%%%%%%%%%%%%%%%%%%%%%%%%%%%%%
\section{Deep Learning Coordinated  Beamforming} \label{sec:DL_Framework}
%%%%%%%%%%%%%%%%%%%%%%%%%%%%%%%%%%%%%%%%%%%%%%%%%%%%%%%%%%%%%%%%

Machine learning has attracted considerable interest in the last few years, thanks to its ability in creating smart systems that can take successful decisions and make accurate predictions. Inspired by these gains, this section introduces a novel application of machine learning in mmWave coordinated beamforming. We show that leveraging  machine learning tools can yield interesting performance gains that are very difficult to attain with traditional communication systems. In the next subsections, we first explain the main idea of the proposed coordinated deep learning beamforming solution, highlighting its advantages. Then, we delve into a detailed description of the system operation and the machine learning modeling. For a brief background on machine/deep learning, we refer the reader to \cite{Goodfellow-et-al-2016}. 
%Appendix \ref{app:ML}. 

\subsection{The Main Idea}
As discussed in \sref{sec:BL_Comm}, the key challenge in supporting highly-mobile mmWave applications is the large training overhead associated with estimating the large-scale MIMO channel or scanning the large number of narrow beams. An important note about these beam training solutions (and similarly for compressed sensing) is that they normally do not make any use of the past experience, i.e., the previous beam training results. Intuitively, the beam training result is a function of the environment setup (user/BS locations, room furniture, street buildings and trees, etc.). These functions, though, are difficult to characterize by closed-form equations, as they generally convolve many parameters and are unique for every environment setup. 

In this paper, we propose to\textbf{ integrate deep learning models with the communication system design to learn the implicit mapping function relating the environment setup, which include the environment geometry and user location among others, and the beam training results}. 
To achieve that, the main question is how to characterize the user locations and environment setup in the learning models at the BSs? One solution is to rely on the GPS data fed back from the users. This solution, however, has several drawbacks: (i) the GPS accuracy is normally in the order of meters, which may not be reliable for mmWave narrow beamforming, (ii) GPS devices do not work well inside buildings, and therefore will not support  indoor applications, such as wireless virtual/augmented reality. Further, relying only on the user location is insufficient as the beamforming direction depends also on the environment, which is not captured by the GPS data.  
In the proposed solution,  \textbf{the machine learning model uses the uplink pilot signal received at the terminal BSs with only \textit{omni} or quasi-omni beam patterns to learn and predict the best RF beamforming vectors}. 
%Note that these received pilot signals are the result of the propagation, reflection, and diffraction of the transmitted signal at the different elements of the environment. Therefore, these pilots, which are received \textit{jointly} at the different BSs, draw an RF signature of the environment and the user/BS locations --- the signature we need to learn the beamforming directions. 
%
Note that these received pilot signals at the BSs are the results of the interaction between the transmitted signal from the user and the different elements of the environment through propagation, reflection, and diffraction. Therefore, these pilots, which are received \textit{jointly} at the different BSs, draw an RF \textit{signature} of the environment and the user/BS locations --- the signature we need to learn the beamforming directions.

This proposed coordinated deep learning solution operates in two phases. In the first phase (learning), the deep learning model monitors the beam training operations and learns the mapping from the omni-received pilots to the beam training results. In the second phase (prediction), the system relies on the developed deep learning model to predict the best RF beamforming using only the omni-received pilots, totally eliminating the need for beam training. This solution, therefore, achieves multiple important gains in the same time. First, it does not need any special resources for learning, such as GPS data, as the deep learning model learns how to select the beamforming vectors directly from the received uplink pilot signal. Second, since the deep learning model predicts the best RF beamforming vectors using only omni-received uplink pilots, the proposed solution has negligible training overhead and can efficiently support highly-mobile mmWave applications, as will be shown in \sref{sec:Results}. It is worth noting here that while combining the uplink training signal with omni patterns penalizes the receive SNR, we show in \sref{subsec:Comm_Impact} that this is still sufficient to efficiently train the learning model with reasonable uplink transmit power. Another key advantage of the proposed system operation is that the deep learning model does not need to be trained before deployment, as it learns and adapts to any environment, and can support both LOS and NLOS scenarios. Further, as we will see in \sref{sec:Results}, the deep learning model learns and memorizes the different scenarios it experiences, such as different traffic patterns, which enables it to become more robust over time. Finally, since the proposed deep learning model is integrated with the baseline coordinated beamforming solution, the resulting system inherits the coverage, reliability, and latency gains discussed in \sref{sec:Problem}. 
%We will discuss all these gains in detail throughout the rest of the paper. 

%%%%%%%%%%%%%%%%%%%%%%%%%%%%%%%%%%%%%%%%%%%%%%%%%%%%%%%%%%%%%%%%%%%%%%%%%%%%%%%%%%
\subsection{System Operation} \label{subsec:Oper}

The proposed deep learning coordinated beamforming integrates machine learning with the baseline beamforming solution in \sref{sec:BL_Comm} to reduce the training overhead and achieves high effective achievable rates. This integrated system operates in two phases, namely the online learning and the deep learning prediction phases depicted in Figures \ref{fig:SyS_Op} and \ref{fig:SyS_Opx}. Next, we explain the two phases in detail.

\begin{figure}[t]
	\centering
	\subfigure[center][{Online Learning Phase}]{
		\includegraphics[scale=1]{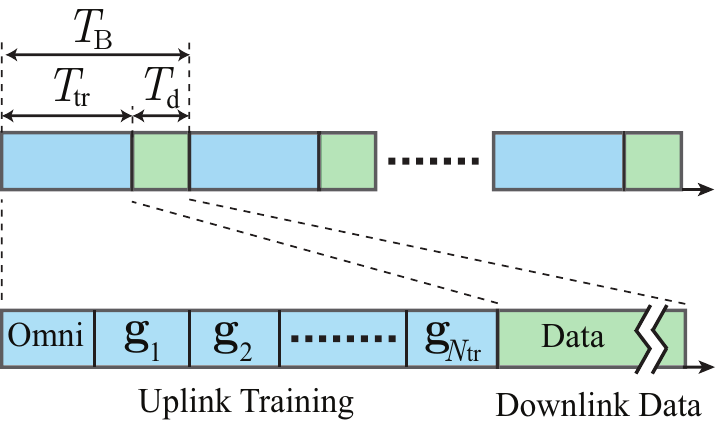}
		\label{fig:Op1}}
	\subfigure[center][{Deep Learning Prediction Phase}]{
		\includegraphics[scale=1]{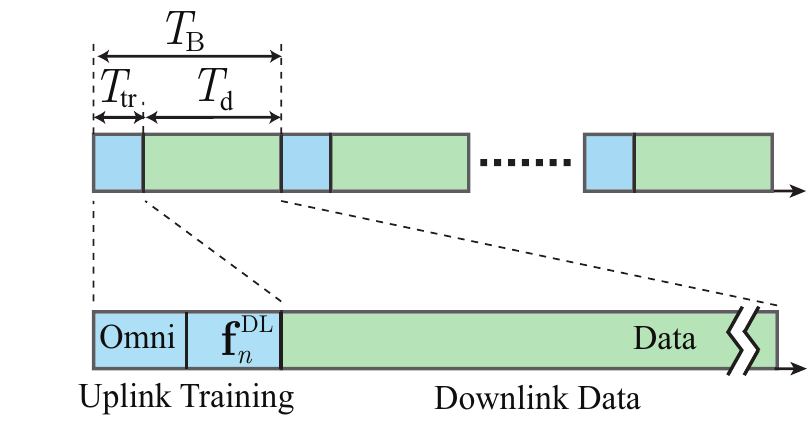}
		\label{fig:Op2}}
	\caption{This figure abstracts the timing diagram of the two phases of the proposed deep learning coordinated beamforming strategy. In the online learning phase, the BSs combine the uplink training pilot using both codebook beams and omni/quasi-omni patterns. In the deep-learning prediction phase, only omni-patterns are used to receive the uplink pilots. }\label{fig:SyS_Op}
\end{figure}

\textbf{Phase 1: Online learning phase:} 
In this phase, the machine learning model monitors the operation of the baseline coordinated beamforming system and trains its neural network. Specifically, for every beam coherence time $T_\mathrm{B}$, the user sends $N_\mathrm{tr}+1$ repeated uplink training pilot sequences $\left\{s_k^{\mathrm{pilot}}\right\}_{k=1}^K$. 
Similar to the baseline solution explained in \sref{subsec:BL}, every BS switches between its $N_\mathrm{tr}$ RF beamforming beams in the codebook $\boldsymbol{\cF}^\mathrm{RF}$ such that it combines every received pilot sequence with a different RF beamforming vector. The only difference is that every BS $n$ will also receive one additional uplink pilot sequence using an omni (or quasi-omni) beam, $\bg_0$, as depicted in \figref{fig:Op1x}, to obtain the received signal 
\begin{equation} \label{eq:o_rec}
r_{k,n}^{\text{omni}}= \bg_{0}^T \bh_{k,n} s^\mathrm{pilot}_k+ \bg_0^T \bv_{k,n}, \hspace{50pt} k=1, 2, ..., K.
\end{equation}

\begin{figure}[t]
	\centering
	\subfigure[center][{Online Learning Phase}]{
		\includegraphics[scale=1.05,trim={21pt 13pt 0pt 0},clip]{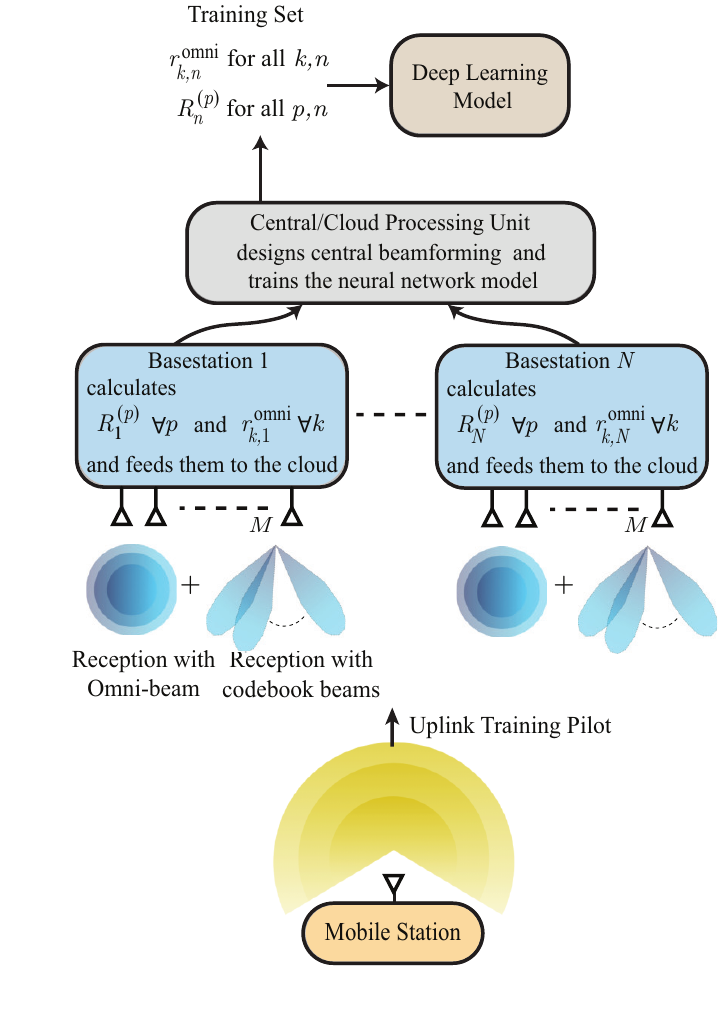}
		\label{fig:Op1x}}
	\subfigure[center][{Deep Learning Prediction Phase}]{
		\includegraphics[scale=1.05,trim={10pt 13pt 5pt 0},clip]{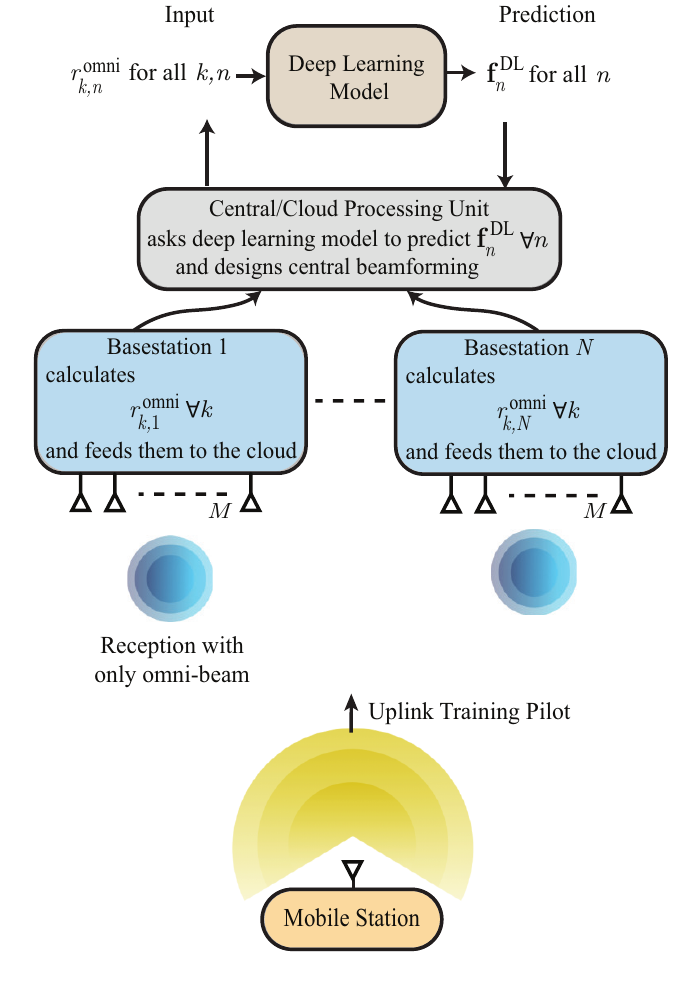}
		\label{fig:Op2x}}
	\caption{This figure illustrates the system operation of the proposed deep-learning coordinated beamforming solution, which consists of two phases. In the online learning phase, the deep-learning model leverages the signals received with both omni and codebook beams to train its neural network. In the deep learning prediction phase, the deep-learning model predicts the BS RF beamforming vectors relying on only omni-received signals, requiring negligible training overhead.}\label{fig:SyS_Opx}
\end{figure}

\noindent The combined signals $r_{k,n}^{\text{omni}}, r_{k,n}^{(p)}, p=1, ..., N_\mathrm{tr}, \forall k$ will be fed back from all the BS terminals to the cloud. The cloud performs two tasks. First, it selects the downlink RF beamforming vector for every BS according to \eqref{eq:BL_RF} and the baseband beamformers as in \eqref{eq:Sol1}, which is similar to the baseline solution in \sref{subsec:BL}. Second, it feeds the machine learning model with (i) the omni-received sequences from all the BSs $r_{k,n}^{\text{omni}}, \forall n$ which represent the inputs to the deep learning model, and (ii) the achievable rate of every RF beamforming vector $R^{(p)}_n, n=1,...,N_\mathrm{tr}$ defined as
\begin{equation} \label{eq:RF_Rate}
R^{(p)}_{n}= \frac{1}{K} \sum_{k=1}^K \log_2\left(1+  \mathsf{SNR}  \left|\bh_{k,n}^T \bg_p \right|^2 \right),
\end{equation}
which represent the desired outputs from the machine learning model, as will be described in detail in \sref{subsec:ML_model}.  The deep learning model is, therefore, trained online to learn the implicit relation between the OFDM omni-received signals captured jointly at all the BSs, which represent a defining signature for the user location/environment, and the rates of the different RF beamforming vectors.  Once the model is trained, the system operation switches to the second phase --- deep learning prediction. It is important to note here that using omni patterns at the BSs during the uplink training reduces the receive SNR compared to the case when combining the received signal with narrow beams. We show in \sref{subsec:Comm_Impact}, though, that this receive SNR with omni patterns is sufficient to efficiently train the neural networks under reasonable assumptions on the uplink training power.

\textbf{Phase 2: Deep learning prediction phase :}  
In this phase, the system relies on the trained deep learning model to predict the RF beamforming vectors based on \textit{only} the omni-received signals captured at the BS terminals. Specifically, at every beam coherence time, $T_\mathrm{B}$, the user transmits an uplink pilot sequence $\left\{s_k^{\mathrm{pilot}}\right\}_{k=1}^K$. The BS terminals combine the received signals using the omni (or quasi-omni) beamforming patterns $\bg_0$ used in the online learning phase. This constructs the combined signals $r_{k,n}^{\text{omni}}$ which are fed back to the cloud processor, as depicted in \fig{fig:Op2x}. Using these omni combined signals $r_{k,n}^{\text{omni}} \forall n, \forall k$, the cloud then asks the trained deep learning model to predict the best RF beamforming vector $\bff_n^\mathrm{DL}$ that maximizes the achievable rate in \eqref{eq:RF_Rate} for every BS $n$. Finally, the predicted RF beamforming vectors $\bff_n^\mathrm{DL}, n=1, ..., N$ are used by the BS terminals to combine the uplink pilot sequence, and to estimate the effective channels $\bh_{k,n}^T \bff_n^\mathrm{DL}, \forall k, n$, which are used to construct the cloud baseband beamforming vectors according to \eqref{eq:Sol1}. 

In the deep learning prediction phase, the system effective achievable rate $R^\mathrm{DL}_\mathrm{eff}$ is given by 
\begin{equation} \label{eq:EfFR_DL}
R_\mathrm{eff}^\mathrm{DL}=\left(1-\frac{2 T_\mathrm{p}}{T_\mathrm{B}}\right) \frac{1}{K} \sum_{k=1}^K \log_2\left(1+  \mathsf{SNR} \sum_{n=1}^N \left|\bh_{k,n}^T \bff^\mathrm{DL}_n \right|^2  \right),
\end{equation}
where the training time $2  T_\mathrm{p}$ represents the time spent for the uplink training of the omni pattern $\bg_0$ and the predicted beam $\bff_n^\mathrm{DL}$, each requiring one beam training pilot sequence time, $T_\mathrm{p}$. Note that we neglected the processing time of executing the deep learning model, as it is normally one or two orders of magnitude less than the over-the-air beam training time, $T_\mathrm{p}$.  It is also worth mentioning that, in general, the deep learning model can predict the best $N_\mathrm{B}$ beams for every BS to be refined in the uplink training, instead of just predicting the best beam, $\bff^\mathrm{DL}_n$. In this case, the training overhead will be $(N_\mathrm{B}+1) T_\mathrm{p}$, which will still be much smaller than the baseline training overhead, as $N_\mathrm{B}$ should typically be much smaller than $N_\mathrm{tr}$.

An important question is when will the system switch its operation from the first phase (learning) to the second phase (prediction)? During the learning phase, and thanks to the proposed system design, the cloud processor can keep calculating both the effective achievable rate of the baseline solution $R_\mathrm{eff}^\mathrm{BL}$, and the \textit{estimated} effective rate of the learning phase $R_\mathrm{eff}^\mathrm{DL}$. The system can then switch to the deep learning prediction phase when $R_\mathrm{eff}^\mathrm{DL} > R_\mathrm{eff}^\mathrm{BL}$. This also results in an overall effective achievable rate of $\max (R_\mathrm{eff}^\mathrm{BL}, R_\mathrm{eff}^\mathrm{DL})$. Note that this result implies that the deep learning model will only be leveraged when it can achieve a better rate than the baseline solution and that it has almost no cost on the system performance.  Finally, we assume for simplicity that the system will completely switch to the second phase after the deep learning model is trained. In practice, however, the system should periodically switch back to the online learning phase to ensure updating the learning model with any changes in the environment. Designing and optimizing this mixed system operation for  time-varying environment models is an interesting future research direction. 

%%%%%%%%%%%%%%%%%%%%%%%%%%%%%%%%%%%%%%%%%%%%%%%%%%%%%%%%%%%%%%%%%%%%%%%%%%%%%%%%%%
\subsection{Machine Learning Modeling} \label{subsec:ML_model}

In this subsection, we describe the different elements of the proposed machine learning model: (i) the input/output representation and normalization, (ii) the neural network architecture, and (iii) the adopted deep learning model. 
It is worth mentioning that the machine learning model presented in this section is just one possible solution for the integrated communication and learning system  proposed in \sref{subsec:Oper}, with no optimality guarantees on its performance or complexity. Developing other machine learning models with higher performance and less complexity is an interesting and important future research direction.

\textbf{Input representation and normalization:} As discussed in \sref{subsec:Oper}, the proposed deep learning coordinated beamforming solution relies on omni (or quasi-omni) received signals to predict distributed beamforming directions. Based on that, we propose to define the inputs to the neural network model as the OFDM omni-received sequences, $r_{k,n}^{\text{omni}}$, collected from the $N$ BSs. Since the sparse mmWave channel is highly correlated in the frequency domain \cite{Venugopal2017}, we will only consider a subset of the OFDM symbols for the inputs of the learning model. For simplicity, we will set the inputs of the model to be equal to the first $K_\mathrm{DL}$ samples, $r_{k,n}^{\text{omni}}, k=1,..., K_\mathrm{DL}$ of the K-point OFDM symbol. Note that inputing the raw data directly to the neural network without extracting further features is motivated by the ability of deep neural networks in learning the hidden and relevant features of the inputs \cite{Goodfellow-et-al-2016}. Finally, We represent every received signal $r_{k,n}^{\text{omni}}$ by two inputs, $\Re\left\{r_{k,n}^{\text{omni}}\right\}, \Im\left\{r_{k,n}^{\text{omni}}\right\}$, carrying the real and imaginary components of $r_{k,n}^{\text{omni}}$. Therefore, the total number of inputs to the learning model is $2 K_\mathrm{DL} N$, as depicted in \figref{fig:ML_model}.

\begin{figure}[t]
	\centerline{
		\includegraphics[width=.9\columnwidth]{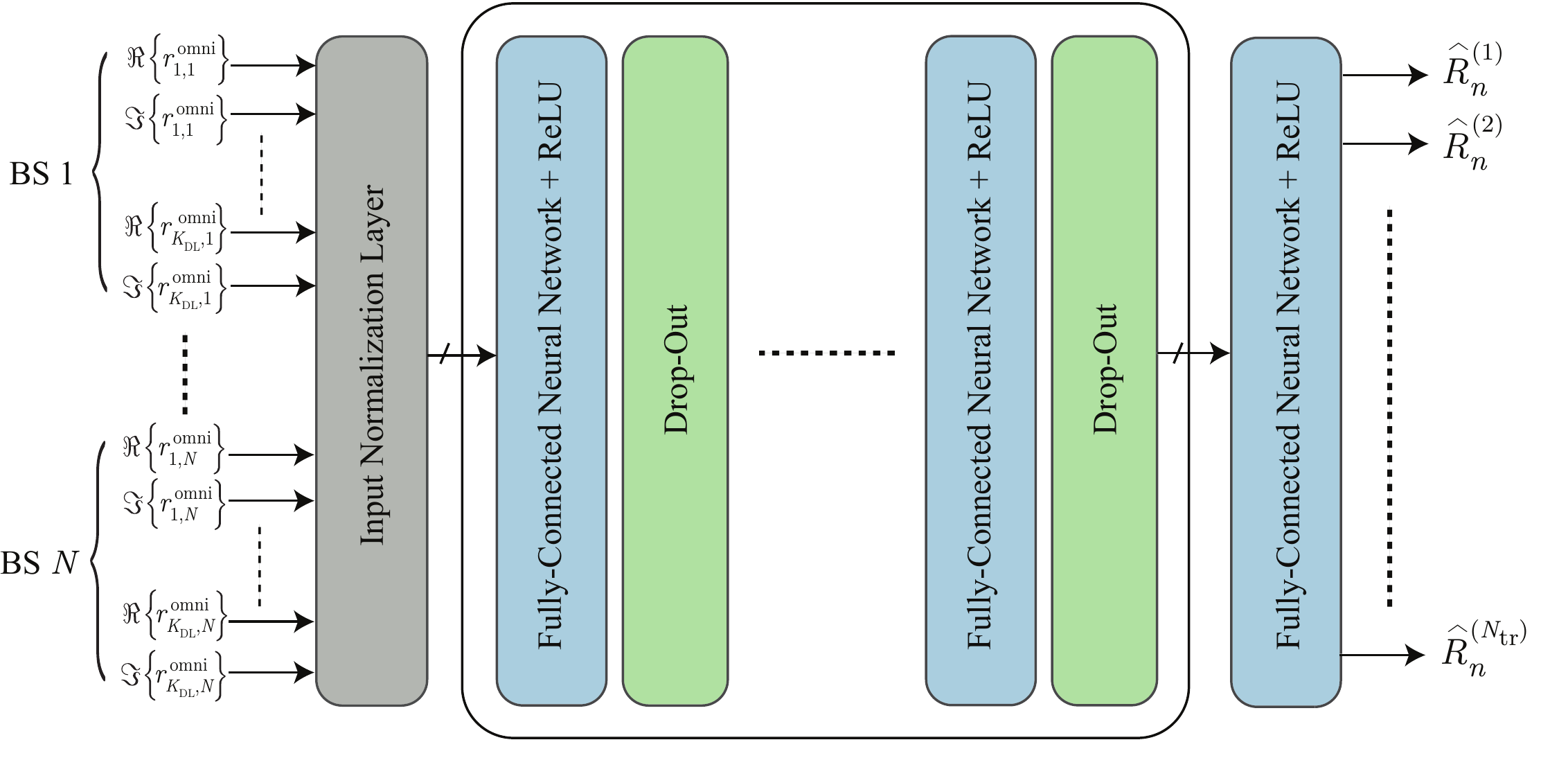}
	}
	\caption{A block diagram of the proposed machine learning model for the $n$th BS. The model relies on the OFDM omni-received sequences from the $N$ BSs to predict the $n$th BS achievable rate with every RF beamforming codeword.}
	\label{fig:ML_model}
\end{figure}

% Normalization

Normalizing the inputs of the neural network normally allows using higher learning rates and makes the model less affected by the initialization of the neural network weights and the outliers of the training samples \cite{Ioffe2015}. For our application, there are four main approaches in normalizing the model inputs: (i) per-carrier per-BS normalization, where we independently normalize every received signal $r_{k,n}^{\text{omni}}$ of every carriers and BS, (ii) per-BS normalization, where we apply the same normalization/scaling to all the carriers of the BS, but independently from the other BSs, (iii) per-sample normalization, where the $2 K_\mathrm{DL} N$ inputs of every learning sample are subject to the same normalization/scaling, and (iv) per-dataset normalization, where we only scale the whole dataset by a single factor.   

% Differentiate between scaling and normalization? 

In our coordinated beamforming application, the correlation between the received signals at the same BS may carry important information that will be lost if a per-carrier normalization is adopted. Similarly, the correlation between the signals received at different BSs from the same user may carry some information about the relative location and multi-path patterns for this user and every BS. This information will be distorted when using a per-BS normalization. Further, the correlation between the joint multi-path patterns at the $N$ BSs for different user locations may carry relevant information, which will be lost when using a per-sample normalization. Therefore, it is intuitive to adopt a per-dataset normalization in our coordinated beamforming application to avoid losing any information that could be useful for the learning model. This intuition is also confirmed by the simulation results in \sref{sec:Results}. In these simulations, we consider a simple per-dataset normalization where all the inputs are divided by a constant scaler $\Delta_{\mathrm{norm}}$, defined as 
\begin{equation}
\Delta_{\mathrm{norm}}=\max_{\substack{ k=1, ..., K_\mathrm{DL}, \\ n=1, ..., N, \\  s=1, ..., S}} I_{k,n,s},
\end{equation}
where $I_{k,n,s}$ denotes the absolute value of the omni-received signal $r_{k,n}^{\text{omni}}$ at the $n$th BS and $k$th subcarrier for the $s$th learning sample.

\textbf{Output representation and normalization:} As shown in \sref{sec:BL_Comm}, separating the BS RF and cloud baseband beamforming  design problems yields low-complexity yet highly-efficient systems, with achievable rates approaching the optimal bound in some important cases. With this motivation, we propose to have $N$ independent deep learning models for the $N$ BSs, where the objective of every model, $n, n= 1, ..., N$, is to predict the best RF beamforming vector $\bff_n^\mathrm{DL} \in \boldsymbol{\cF}_\mathrm{RF}$ with the highest data rate for the $n$th BS. Note that every model, $n$, will still rely on the omni-received sequences from the $N$ BSs to predict the beamforming vectors of BS $n$, as shown in \figref{fig:ML_model}. Further, every deep learning model has $N_\mathrm{tr}=\left|\boldsymbol{\cF}_\mathrm{RF}\right|$ outputs, each representing the predicted rate with one of the $N_\mathrm{tr}$ RF beamforming vectors. 

In the online learning phase, explained in \sref{subsec:Oper}, a new training sample for the deep learning models is generated every beam coherence time, $T_\mathrm{B}$. This training sample for the $n$th BS model consists of (i) the omni-received sequences $r_{k,n}^{\text{omni}}, \forall k, \forall n$, which are the inputs to the deep learning model, and (ii) the achievable rates, $R^{(p)}_{n}, p=1,..., N_\mathrm{tr}$, for the $N_\mathrm{tr}$ RF beamforming vectors, which represent the \textit{desired} outputs from the model. Note that both the omni-received sequences $r_{k,n}^{\text{omni}}, \forall k, \forall n$ and the achievable rates $R^{(p)}_{n}, \forall p, \forall n$ are constructed during the uplink training phase, as described in \sref{subsec:Oper}. These training samples are used by the cloud to train the deep learning models of the $N$ BSs. For the training of the $n$th BS model, $n, 1, ..., N$, the desired outputs $R^{(p)}_{n}, p=1,..., N_\mathrm{tr}$ of every training sample are normalized as 
\begin{equation}
\overline{R}^{(p)}_{n}=\frac{{R}^{(p)}_{n}}{\max_{p}{R}^{(p)}_{n}}.
\end{equation}
\noindent  The objective of this per-sample normalization is to regularize the deep neural network and make sure it does not learn only from the samples with higher data rates (higher output values). This is particularly important for mmWave systems where some user locations have LOS links (with high data rates) while others experience non-LOS connections (with much lower data rates). In this case, if the training samples are not normalized, the neural network model may learn only from the LOS samples, as will be illustrated in \sref{sec:Results}.

\textbf{Neural network architecture:} 
The main objective of this paper is to develop an integrated communication-learning coordinated beamforming approach for highly-mobile mmWave applications. Optimizing the deep neural network model, though, is out of the scope of this paper, and is one of the important future research directions. 
%, as will be discussed in \sref{sec:Future}. 
In this paper, we adopt a simple neural network architecture based on fully-connected layers. As shown in \figref{fig:ML_model}, the neural network architecture consists of $M_\mathrm{Layer}$ fully-connected layers, each with $M_\mathrm{Nodes}$ nodes. The fully-connected layers use rectifier linear units (ReLU) activations \cite{Goodfellow-et-al-2016}. Every fully-connected layer is followed by a drop-out layer to ensure the regularization and avoid the over-fitting of the neural network \cite{Srivastava2014}. The performance of the proposed deep learning coordinated beamforming solution with the adopted neural network architecture as well as comparisons with other network architectures will be discussed in \sref{sec:Results}.

\textbf{Loss function and learning model:}
The objective of the deep learning model is to predict the best RF beamforming vectors with the highest achievable rates for every BS. Therefore, we adopt a \textit{regression} learning model in which the neural network of every model $n, n=1, ..., N$, is trained to make its outputs, $\widehat{R}_{n}^{(p)}, p=1, ..., N_\mathrm{tr}$, as close as possible to the desired normalized achievable rates, $\overline{R}_{n}^{(p)}, p=1, ..., N_\mathrm{tr}$. Note that adopting a regression model enables the neural network to predict not only the best RF beamforming vector, but the second best, third best, etc. --- or generally, the best $N_\mathrm{B}$ RF beams. Formally, the neural network for every model $n$ is trained to minimize the loss function, $L_n\left(\boldsymbol{\theta}\right)$, defined as 
\begin{equation}
L_n\left(\boldsymbol{\theta}\right)=\sum_{p=1}^{N_\mathrm{tr}} \mathsf{MSE} \left(\overline{R}_{n}^{(p)},\widehat{R}_{n}^{(p)}\right),
\end{equation}
where $\mathsf{MSE}\left(\overline{R}_{n}^{(p)},\widehat{R}_{n}^{(p)}\right)$ is the mean-squared-error between $\overline{R}_{n}^{(p)}$ and $\widehat{R}_{n}^{(p)}$, and $\boldsymbol{\theta}$ denotes the set of all the parameters in the neural network. 
Note that the outputs of the learning model, $\widehat{R}_{n}^{(p)}, \forall p$, are functions of the network parameters $\boldsymbol{\theta}$ and the model inputs $r_{k,n}^{\text{omni}}, \forall k, \forall n$. To simplify the notation, though, we dropped these dependencies from the symbol $\widehat{R}_{n}^{(p)}, \forall p$.
		
%%%%%%%%%%%%%%%%%%%%%%%%%%%%%%%%%%%%%%%%%%%%%%%%%%%%%%%%%%%%%%%%%%%%%%%%%%%%%%%%%%
\subsection{Effective Achievable Rate and Mobility Support}

As shown in \sref{sec:BL_Comm}, the achievable rate with the baseline coordinated beamforming solution approaches the optimal bound in some special yet important cases. The challenge with the baseline solution, though, is the requirement of exhaustive beam training which consumes a lot of training resources and significantly reduces the effective achievable rate. For the deep learning coordinated beamforming solution, the learning model is trained to approach the achievable rate of the baseline solution, which is optimal in some cases. Further, it requires only two training resources for the omni pattern and predicted beam training, which makes its training overhead almost negligible. This means that the proposed deep learning coordinated beamforming solution, when efficiently trained, can approach the optimal effective achievable rate, $R^\star$, and support highly-mobile mmWave applications, as will be shown in the following section.

%%%%%%%%%%%%%%%%%%%%%%%%%%%%%%%%%%%%%%%%%%%%%%%%%%%%%%%%%%%%%%%%%%%%%%%%%%%%%%%%%%%%%%%%%%%%%%%%%%%%%%%%%
\section{Simulation  Results} \label{sec:Results}
%%%%%%%%%%%%%%%%%%%%%%%%%%%%%%%%%%%%%%%%%%%%%%%%%%%%%%%%%%%%%%%%

In this section, we evaluate the performance of the proposed  coordinated deep-learning beamforming solution, and illustrate its ability to support highly-mobile mmWave applications. First, we present the considered simulation setup in \sref{subsec:Setup}. Then, we show the capability of the proposed deep learning solution in predicting the beamforming directions and approach the optimal \textit{effective} achievable rate in \sref{subsec:Main_Result}. In Sections \ref{subsec:Comm_Impact} - \ref{subsec:ML_impact}, we study the impact of the main communication and machine learning parameters on the system performance. Finally, Sections \ref{subsec:Adapt} - \ref{subsec:sync} investigate several important aspects of the integrated communication/learning beamforming system such as its ability to adapt with the environment, its sensitivity to BSs synchronization, and its performance with untrained scenarios.

%%%%%%%%%%%%%%%%%%%%%%%%%%%%%%%%%%%%%%%%%
\subsection{Simulation Setup} \label{subsec:Setup}

This section describes in detail the various aspects of the considered simulation setup including the communication system/channel models, the machine learning model, and the simulation scenarios. While the coordinated beamforming strategies proposed in this paper are general for indoor/outdoor applications, we focus in these simulation results on the vehicular application, which is one important use case for 5G cellular systems \cite{Choi2016,5GPPP}. 

\begin{figure}[t] 
	\centering
	\subfigure[center][]{
		\includegraphics[width=.47\columnwidth]{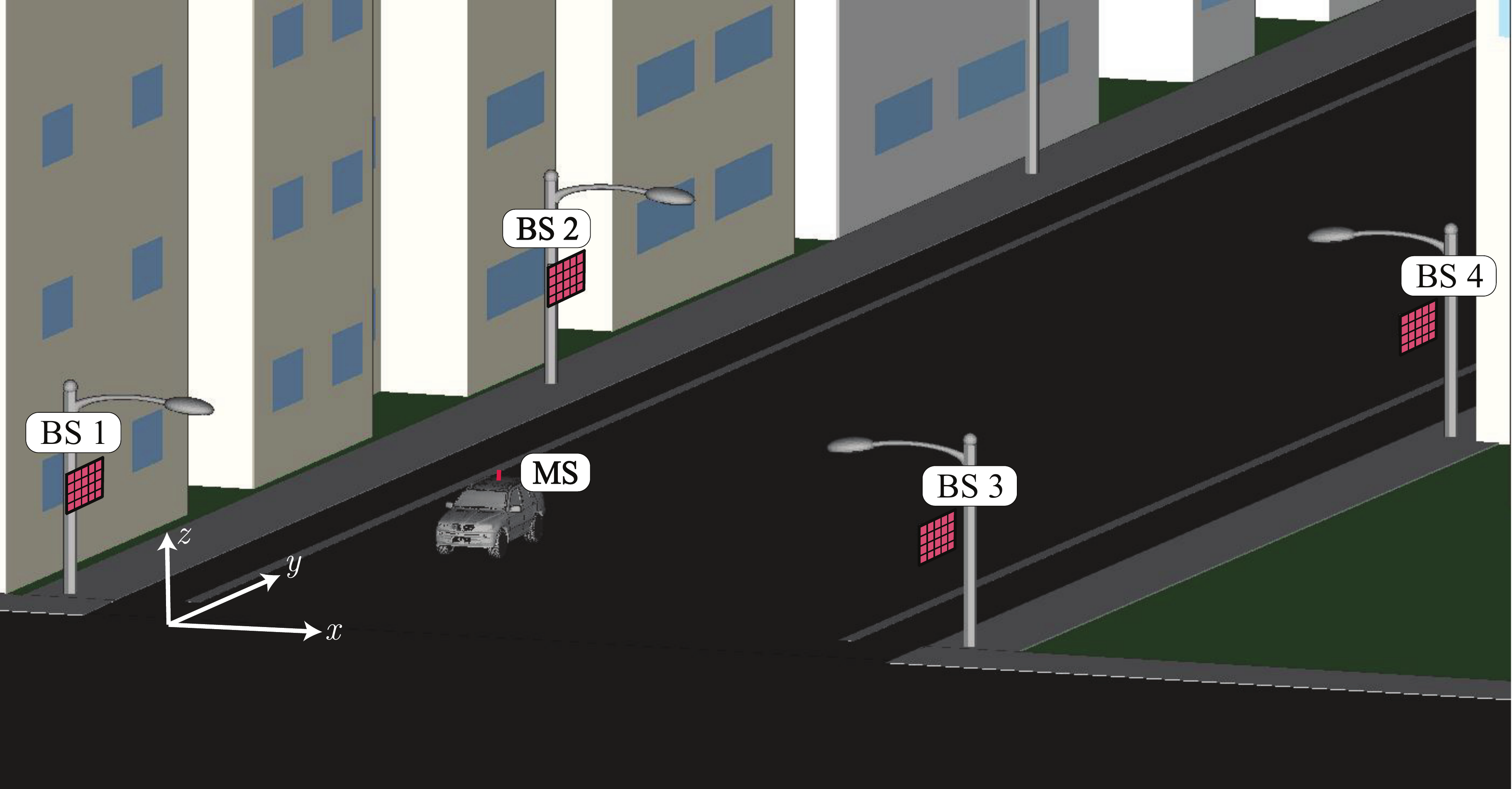}
		\label{fig:setup1}}
	\subfigure[center][]{
		\includegraphics[width=.47\columnwidth]{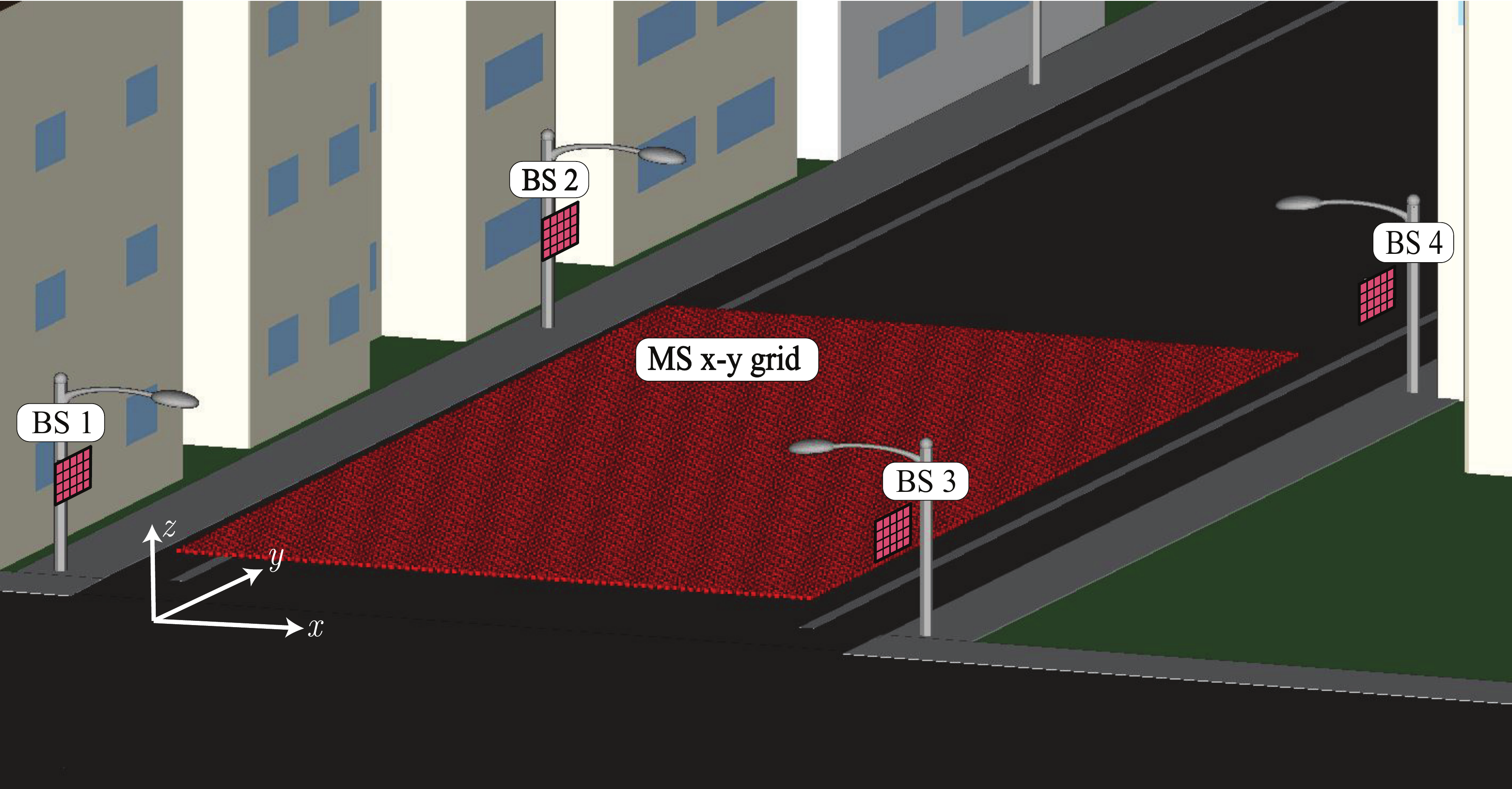}
		\label{fig:setup2}}
	\caption{Figure (a) illustrates the considered street-level simulation setup where 4 BSs, each has UPA, are serving one single-antenna vehicular mobile user. Figure (b) shows the rectangular x-y grid which represents the candidate locations of the mobile user antenna. }
	\label{fig:setup_LOS}
\end{figure}

\noindent \textbf{System Setup and Channel Generation:} 
We adopt the mmWave system and channel models in \sref{sec:Model}, where a number of BSs are simultaneously serving one mobile user over the $60$ GHz band. Since the proposed   deep-learning coordinated  beamforming approach relies on learning the correlation between the transceiver locations/environment geometry and the beamforming directions, it is important to generate realistic data for the channel parameters (AoAs/AoDs/pathloss/delay/etc.).  With this motivation, our simulations use the commercial ray-tracing simulator, Wireless InSite \cite{Remcom}, which is widely used in mmWave research \cite{Yang2006,Va2017a,Li2015a}, and is verified with channel measurements \cite{Wu2016,Li2015b}. In the following points, we summarize the environment/system setup and channel generation. 
\begin{itemize}
	\item \textbf{Environment setup:}
	We consider the system model in \sref{sec:Model} in a street-level environment, where $N=4$ BSs are installed on $4$ lamp posts to simultaneously serve one vehicular mobile user, as depicted in \figref{fig:setup1}.  The 4 lamp posts are located on the corners of a rectangle, with $60$m distance between the lamp posts on each side of the street (along the y-axis), and $50$m distance between the lamp posts across the street (along the x-axis). In the ray-tracing, we use ITU 60 GHz 3-layer dielectric material for the buildings,  ITU 60 GHz single-layer dielectric for the ground, and ITU 60 GHz glass for the windows. This ensures that the important ray-tracing parameters, such as the reflection and penetration coefficients, accurately model the mmWave system operational frequency.

	\item \textbf{Base stations setup:}
	Each BS is installed on one lamp post at height 6 m, and has a uniform planar array (UPA) facing the street, i.e., on the y-z plane. Unless otherwise mentioned, the BS UPAs consist of 32 columns and 8 rows resulting in a total of $M=256$ antenna elements, and use $30$dBm transmit power. Adopting the system model in \sref{sec:Model}, the BSs are assumed to be connected with a central processing via error-negligible delay-negligible links. In practice, this can be realized using optical fiber links connecting the four BSs together, with one of them hosting the central processor.

	\item \textbf{mobile user setup:}
	The vehicular mobile user has a single antenna that is deployed at a height of 2m. We show the car in \figref{fig:setup1} only for illustration. This car, though, is not modeled in the ray-tracing simulations, which only consider the mobile user antenna. At every beam coherence time, the location of the mobile user antenna is randomly selected from a uniform x-y grid of candidate locations, as depicted in \figref{fig:setup2}. The x-y rectangular grid has dimensions $40$m $\times 60$m with a resolution of $0.1$ m, i.e., a total of 240 thousand points. This x-y grid shares the same center with the rectangle defined by the 4 BSs.  During the uplink training, the MS is assumed to use $30$dBm transmit power. 

	\item \textbf{Ray-tracing based channel generation:}
	In our simulations, we adopt the frequency-selective \textit{geometric} channel model in \sref{sec:ChModel}. For this model, the important question is how to generate the channel parameters, such as the AoAs, AoDs, path gains and delays of each ray. We normally resort to stochastic models in generating these parameters \cite{Alkhateeb2016c,HeathJr2016,Sohrabi2015}. In this paper, though, the key idea is to leverage the deep neural network power in learning the mapping between the omni-received multi-path signatures and the beamforming directions. This implicitly relies on learning the underlying environment geometry and the interplay between this geometry and the transmitter/receiver locations. Therefore, it is crucial to generate realistic channel parameters that correspond to real environment geometry. This is the main motivation for using ray-tracing in generating the channel parameters. 
	
	In the Wireless InSite ray-tracing \cite{Remcom}, we use the X3D model with Shooting and Bouncing Ray (SBR) tracing mode. In this mode, the simulator shoots hundreds of rays from the transmitters and select the ones that find paths to the receiver for which it generates the key parameters (AoAs/AoDs/etc.). Considering the ray-tracing channel parameters for the strongest 25 paths, which normally have power gap more than 20dB, we construct the channel matrix between each BS and mobile user using MATLAB, according to \eqref{eq:d-channel}. The considered setup adopts  an OFDM system of size $K=1024$. Note that for every  candidate user location in the x-y grid, we generate $4$ channel vectors which correspond to the channels between this user and the 4 BSs. 
\end{itemize}

\noindent \textbf{Coordinated Beamforming:}
In the simulation results, the beamforming vectors are constructed as described in Sections \ref{sec:BL_Comm}-\ref{sec:DL_Framework}. At every beam coherence time, a new user location is selected, and the channel vectors $\bh_{k,n} \forall k$ are constructed based on the parameters generated from the ray-tracing simulations as described earlier in this section. For the baseline coordinated beamforming, we first simulate the uplink beam training at each BS $n$ by calculating $|\bh_{k,n}^T \bg_p| \forall k$ for all the beamforming vectors $\bg_p$ in the codebook $\boldsymbol{\cF}$. Then, the best RF beamforming vector for every BS is determined based on \eqref{eq:BL_RF}. Finally, the effective achievable rate  is calculated according to  \eqref{eq:EfFR_BL}. In these simulations, we consider an oversampled beamsteering codebook of $N_\mathrm{tr}=M_y M_z N_{\mathrm{OS},y} N_{\mathrm{OS},z}$ beams, with $M_y, M_z$ denoting the number of columns and rows of the BSs UPAs, and $N_{\mathrm{OS},y}, N_{\mathrm{OS},z}$ defining the oversampling factors in the azimuth and elevation  directions. The $p$th beamforming vector in this codebook is expressed as $\bg_p=\ba^*(\overline{\theta}_p, \overline{\phi}_p), p=1, ..., N_\mathrm{tr}$, where $\ba(\overline{\theta}_p, \overline{\phi}_p)$ is the UPA array steering vector with the quantized angles $\overline{\theta}_p, \overline{\phi}_p$. 

The simulation of the deep-learning coordinated beamforming approach is similar to the baseline coordinated beamforming with the following extra steps. First, at every beam coherence time, $T_\mathrm{B}$, i.e., a new user location, in addition to calculating $|\bh_{k,n}^T \bg_p| \forall k$ for all the beams, we also calculate the omni-received sequences $r_{k,n}^{\text{omni}} \forall k$ in \eqref{eq:o_rec}. To do that, we consider the signal received by only the first antenna element, which is equivalent to adopting a beamforming vector $\bg_0=[1, 0, ..., 0]$ in \eqref{eq:o_rec}. For the noise term in \eqref{eq:o_rec}, we add random noise samples taken from $\mathcal{N}_\mathbb{C}(0,N_0)$ with the noise power $N_0$ corresponding to $1$ GHz system bandwidth and $5$ dB noise figure. The omni received sequence from the $N$ BSs and the rate corresponds to every BF vector, calculated based on \eqref{eq:RF_Rate}, form one data point for the machine learning model. By randomly picking $N_\mathrm{DL}$ user locations, we  build an $N_\mathrm{DL}$-point dataset for the machine learning model. In the second phase of the deep-learning coordinated beamforming approach, we simulate the uplink training by only calculating the omni-received sequence $r_{k,n}^{\text{omni}} \forall k$. We then use the machine learning model to predict the best RF beamforming vector $\bff_n^\mathrm{DL}$ for every BS n. Finally, the effective achievable rate is calculated using \eqref{eq:EfFR_DL}.
\begin{figure}[t]
	\centerline{
		\includegraphics[width=.75\columnwidth]{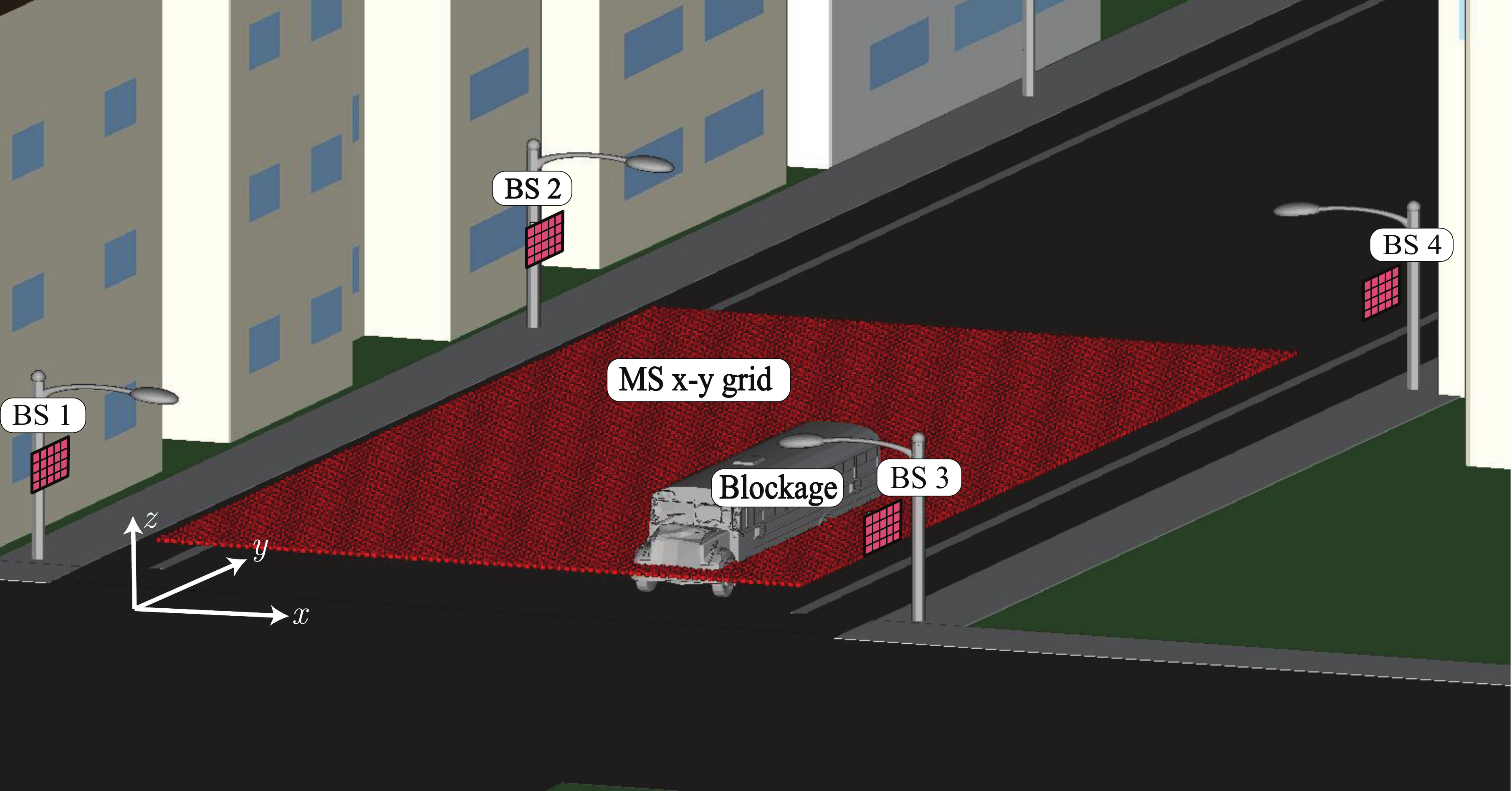}
	}
	\caption{This figure illustrates the non-LOS setup where a bus is blocking the LOS path between most of the possible locations for the mobile user antenna and the UPA of the 3rd BS.}
	\label{fig:setup_NLOS}
\end{figure}

\noindent \textbf{Machine Learning Model:} 
We consider the deep learning model described in detail in \sref{subsec:ML_model}. The neural network model of every BS has $2 N K_\mathrm{DL}$ inputs, which are the the real and imaginary components of the omni-received sequences $r_{k,n}^{\text{omni}}, k=1, ..., K_\mathrm{DL}$ of the $N$ BSs, and $N_\mathrm{tr}$ outputs, which represent the achievable rates $R_n^{\left(p\right)}, \forall p$ of the RF candidate beamforming vectors. Unless otherwise mentioned, the neural network model has $6$ fully connected layers, each of $2 N K_\mathrm{DL} = 512$ nodes, i.e., with $K_\mathrm{DL} =64$. The fully-connected layers use ReLU activation units and every layer is followed by a drop-out regulation layer of dropout rate $.5 \%$. For training the model, we use a dataset with a maximum size of $N_\mathrm{DL}=240$ thousand samples and a batch size of $100$. In the deep learning experimental work, we used the \emph{Keras} libraries \cite{Branchaud-Charron} with a \emph{TensorFlow} \cite{Abadi2016} backend.

\noindent \textbf{LOS and NLOS Scenarios:}
In order to evaluate the performance of our proposed deep-learning coordinated beamforming solution in rich mmWave environment with blockage, we consider both LOS and NLOS scenarios in the simulations. Earlier in this section, we described the LOS scenario, which is depicted in \figref{fig:setup_LOS}. The NLOS scenario is similar to the LOS one but with a large bus of dimensions 20 m x 5 m standing in front of BS 3, as shown in \figref{fig:setup_NLOS}. This bus blocks the LOS path between BS 3 and most of the candidate user locations in the x-y grid.  

Next, we evaluate the performance of the proposed deep learning coordinated beamforming solution for various communication and machine learning parameters.

%%%%%%%%%%%%%%%%%%%%%%%%%%%%%%%%%%%%%%%%%
\subsection{Does the System Learn How to Beamform?} \label{subsec:Main_Result}

\begin{figure}[t]
	\centerline{
		\includegraphics[width=.75\columnwidth,height=250pt]{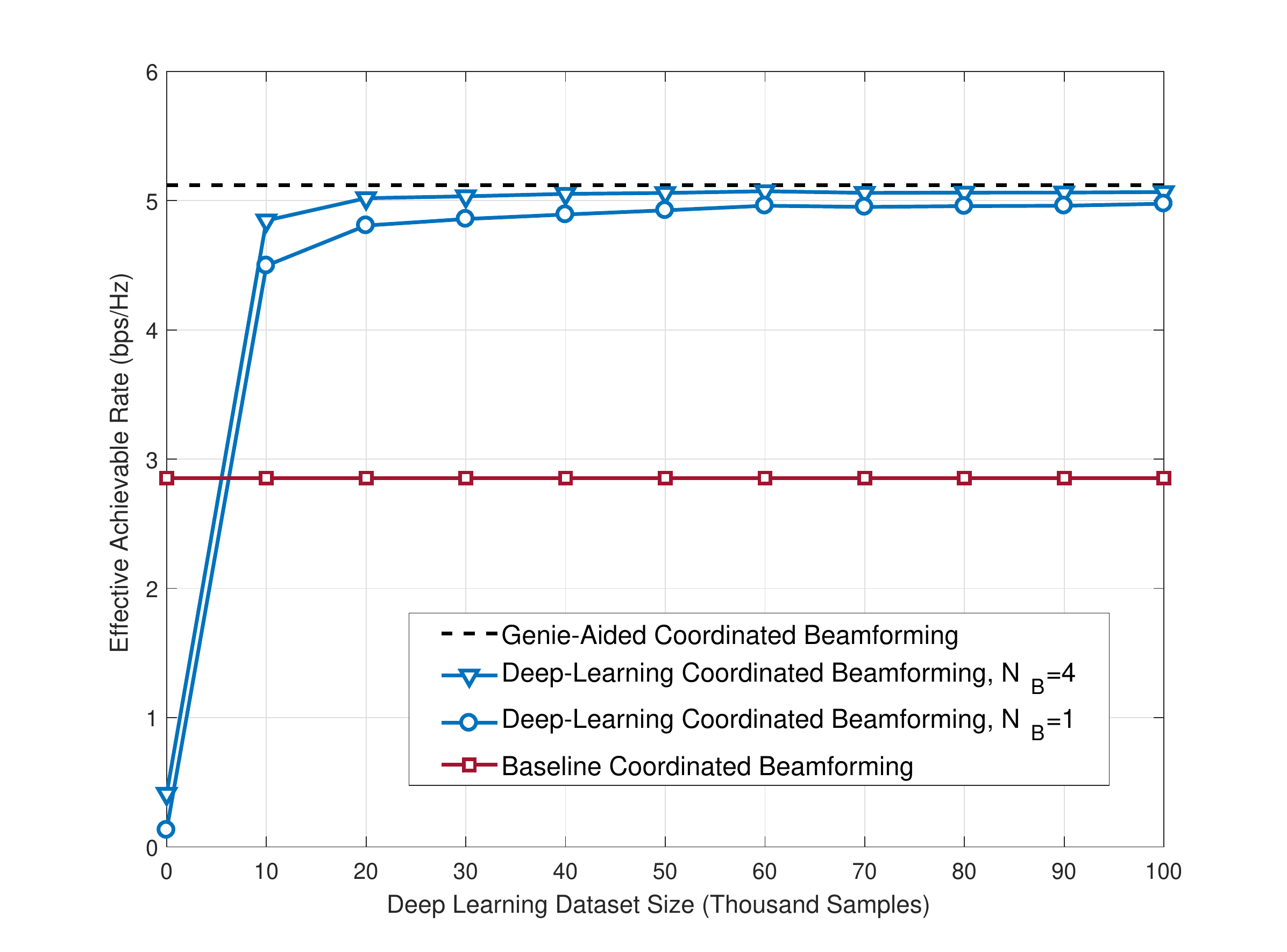}
	}
	\caption{The effective achievable rate of the proposed deep-learning coordinated beamforming solution compared to the baseline coordinated beamforming and the upper bound, $R^\star$. The figure considers a LOS scenario with $4$ BS, each with 32 x 8 UPA, serving one vehicle moving at speed $30$ mph. This figure shows that as we train the neural network model (with more dataset sizes), the performance of the deep-learning coordinated beamforming approaches the optimal effective achievable rate.}
	\label{fig:Main_LOS}
\end{figure}

% -----------------------------------------------
The proposed deep-learning coordinated beamforming solution relies on the ability of  deep neural networks in learning the relation between the multi-path signatures collected jointly at multiple BS locations and the RF beamforming vectors. The first question that we need to address then is whether these networks are successfully  learning how to select the optimal RF beamforming vectors, with the optimality defined according to  \eqref{eq:BL_RF}. To answer this question and to evaluate the quality of this learning, we plot the effective achievable rate of the proposed deep-learning coordinated beamforming for different training dataset sizes in Figures \ref{fig:Main_LOS} and \ref{fig:Main_NLOS}.

\begin{figure}[t]
	\centerline{
		\includegraphics[width=.75\columnwidth,height=250pt]{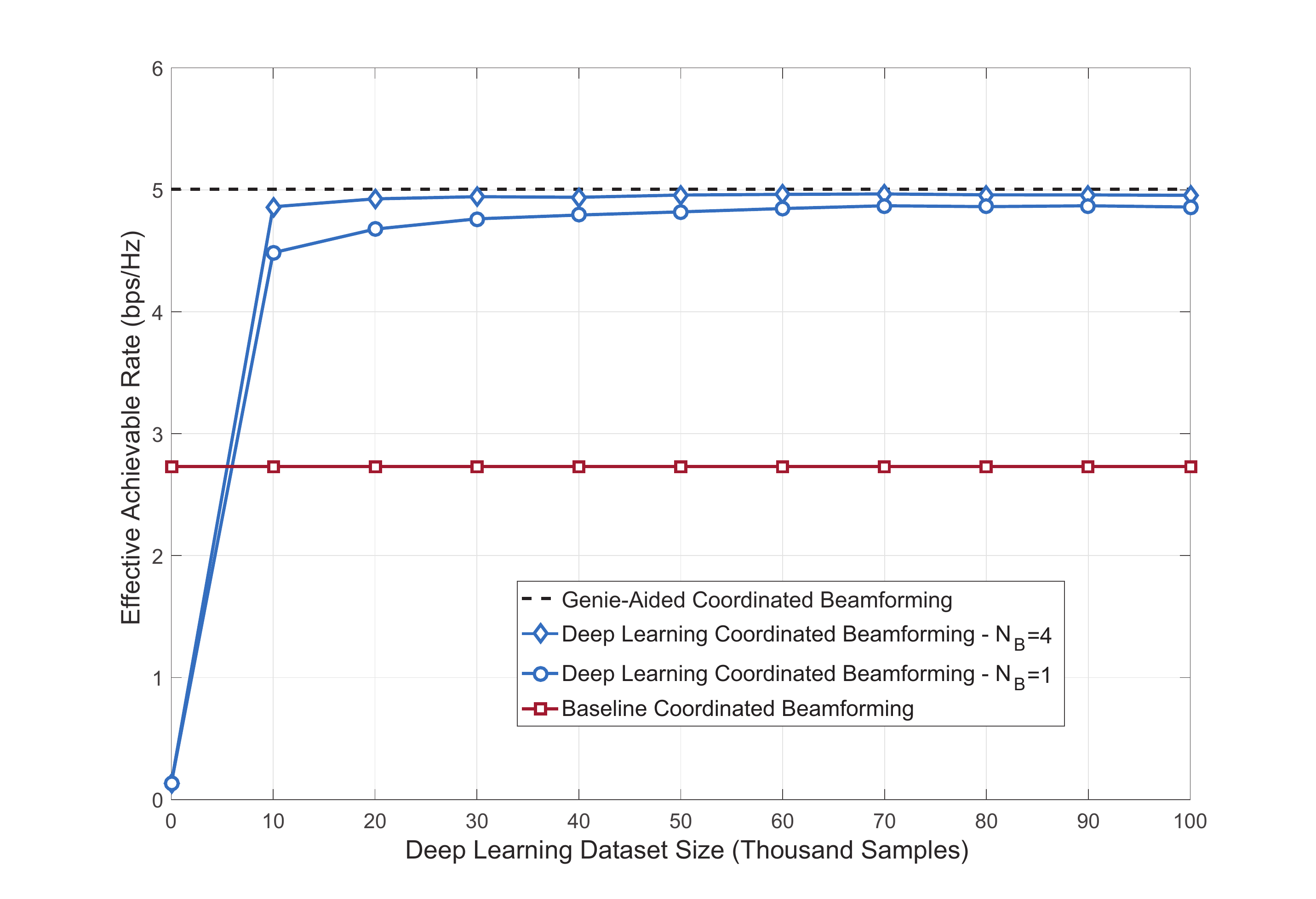}
	}
	\caption{This figure compares the effective achievable rate of the proposed deep-learning and baseline coordinated beamforming solutions with respect to the upper bound, $R^\star$. The figure adopts the NLOS scenario described in \sref{subsec:Setup} where $4$ BS, each with 32 x 8 UPA, serving one vehicle moving with speed $30$ mph. The effective achievable rate of the proposed deep-learning coordinated beamforming approaches the upper bound as larger dataset size is considered, i.e., with more time spent in training the neural network model.} 
	%[XXX This figure will be rechecked]}
	\label{fig:Main_NLOS}
\end{figure}

In \figref{fig:Main_LOS}, we consider the LOS scenario, described in \sref{subsec:Setup}, where $4$ BSs, each with $32 \times 8$ UPA  are simultaneously serving one mobile user, moving with speed $30$ mph. The BSs use beamsteering codebook with oversampling factor of 2 at both the azimuth and elevation directions. For this scenario, we plot the effective achievable rate of the proposed deep-learning coordinated beamforming solution in \figref{fig:Main_LOS} versus the size of the dataset used in training the neural network model. Recall that every point in the training dataset is collected in one beam coherence time, $T_\mathrm{B}$. This means that if the system spent time equals, for example, to $10000  T_\mathrm{B}$ in training its neural network model, then it will be able to predict the beamforming vectors that achieves the effective rate corresponding to the dataset size $10$k samples in \figref{fig:Main_LOS}.  This figure  shows that the effective achievable rate of the proposed deep-learning coordinated beamforming approaches the optimal rate $R^\star$, defined in Lemma \ref{lem:1} with reasonable dataset sizes. \textbf{This means that the neural network model is successfully predicting the best RF beamforming vector, out of 1024 candidate beams, for every BS using multi-path signatures received with only a single antenna (or omni-pattern) at every BS}. This clearly illustrates the ability of the proposed deep-learning based solution in supporting highly-mobile mmWave applications with negligible training overhead.  \figref{fig:Main_LOS} also shows that it is better to select the best $N_\mathrm{B}=4$ beams predicted by the neural network and refine them through beam training, as described in \sref{subsec:Oper}. Finally,  \figref{fig:Main_LOS} illustrates that leveraging deep learning can achieve considerable data rate gains compared to the baseline coordinated beamforming solution.

In \figref{fig:Main_NLOS}, we adopt the NLOS scenario described in \sref{subsec:Setup}, where a large bus is standing in front of BS 3, as shown in \figref{fig:setup_NLOS}. The system, channel, and machine learning models are identical to those adopted in \figref{fig:Main_LOS}. For this NLOS scenario, \figref{fig:Main_NLOS} compares between the effective achievable rate of (i) the developed deep-learning coordinated beamforming strategy with $N_B=1, N_B=4$, (ii) the baseline coordinated beamforming, and (iii) the upper bound, $R^\star$, for different training dataset sizes. \textbf{The result in this figure is very important as it shows that the deep learning model can learn not only LOS beamforming, but also predicting best NLOS beamforming vectors given the joint multi-path signatures}. Note that this is a key advantage of our proposed deep learning solution that relies on the multi-path signature, not on the user location/coordinates, in predicting the beams. If the system relies only on the knowledge of the user location, it will not be able to efficiently predict the beamforming vectors in NLOS scenarios, as the same user location may correspond to different NLOS setups and, consequently, different beamforming vectors.

%%%%%%%%%%%%%%%%%%%%%%%%%%%%%%%%%%%%%%%%%
\subsection{Impact of Communication System Parameters} \label{subsec:Comm_Impact}
The main motivation for the deep-learning coordinated beamforming solution is supporting \textbf{highly-mobile} applications in \textbf{large-array} mmWave systems. In achieving that, our proposed deep learning model makes beamforming predictions based on signals received with only omni or quasi omni antennas, i.e., with \textbf{low-SNR}. In this section, we evaluate the impact of the key system parameters, namely the user mobility, the number of BS antennas, and the uplink transmit power, on the performance of the developed deep-learning coordinated beamforming strategy.

\begin{figure}[t]
	\centerline{
		\includegraphics[width=.75\columnwidth]{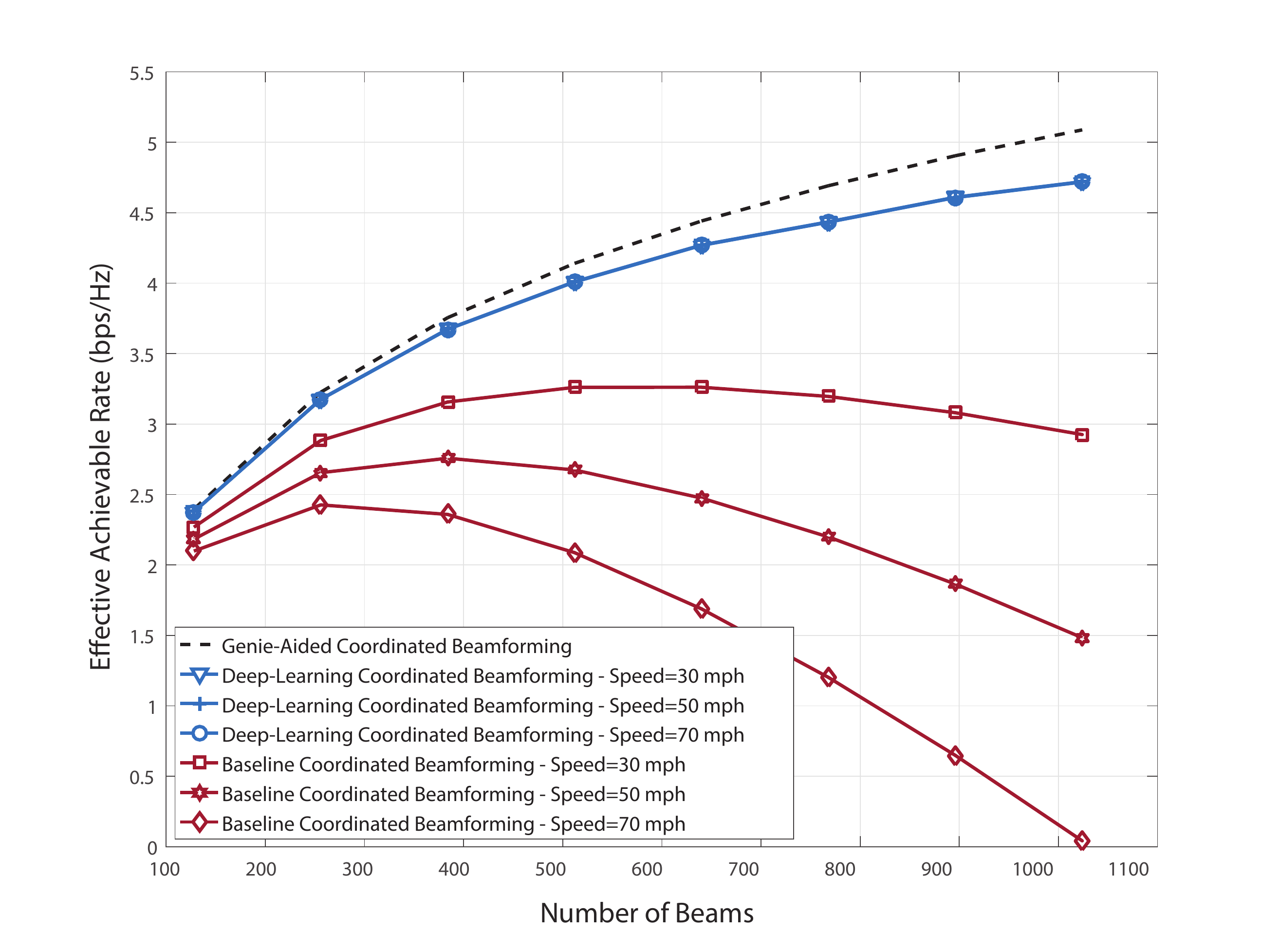}
	}
	\caption{The figure compares the effective achievable rate of the deep-learning coordinated beamforming with the baseline coordinated beamforming and the upper bound, $R^\star$, for different values of BS antennas and user speeds. The deep-learning model was trained with a LOS dataset of size $20$k samples. The performance of the proposed deep-learning coordinated beamforming is nearly as good as the upper bound even for large arrays and highly-mobile users.}
	\label{fig:Speed}
\end{figure}

%---------------------------------------
\textbf{Impact of User Speed and Number of BS antennas:} In \figref{fig:Speed}, we consider the LOS scenario, described in \sref{subsec:Setup}, with 4 BSs serving one mobile user. Each BS is assumed to have a UPA with $M_z=8$ rows, $M_y$ columns, and is using a beamsteering codebook with oversampling factor of 2 in both the  elevation and azimuth directions. In \figref{fig:Speed}, we plot the effective achievable rate of the deep-learning coordinated beamforming solution, the baseline coordinated beamforming, and the upper bound $R^\star$ for different number of BS antennas and user speeds. Recall that the number of beams in the beamsteering codebooks equals 4 (the overall azimuth/elevation oversampling factor) times the number of antennas. First, consider the baseline coordinated beamforming solution performance in \figref{fig:Speed}. As more antennas are deployed at the BSs, the beamforming gain increases but the training overhead also increases, resulting in a trade-off for the effective achievable rate in \eqref{eq:EfFR_BL}. This trade-off defines an optimal number of BS antennas for every user speed (or equivalently beam coherence time), as shown in \figref{fig:Speed}. It is important to note that the performance of the baseline coordinated beamforming solution degrades significantly with increasing the number of BS antennas or the user speed. This illustrates why traditional beamforming strategies are not capable of supporting highly-mobile users in mmWave systems with large arrays.

In contrast, \textbf{the deep-learning coordinated beamforming, which is trained with a dataset of size $20$k samples, achieves almost the same performance of the upper bound for different values of user speeds and BS antennas}. This is thanks to the negligible uplink training overhead using omni patterns. It is worth noting here that while larger arrays may require bigger datasets (longer time) for training the neural network model during the online learning phase, \textbf{the uplink training overhead in the deep learning prediction phase does not depend on the number of antennas as it relies on omni or quasi-omni patterns}. Therefore, once the neural network model is trained, the deep-learning coordinated beamforming solution works efficiently with large antenna arrays. This is a key advantage of our developed deep-learning based solution over traditional mmWave channel training/estimation techniques such as analog beam training \cite{Hur2011,Wang2009} and compressive sensing \cite{Alkhateeb2014,Ramasamy2012,Lee2014}.

\begin{figure}[t]
	\centerline{
		\includegraphics[width=.75\columnwidth,height=260pt]{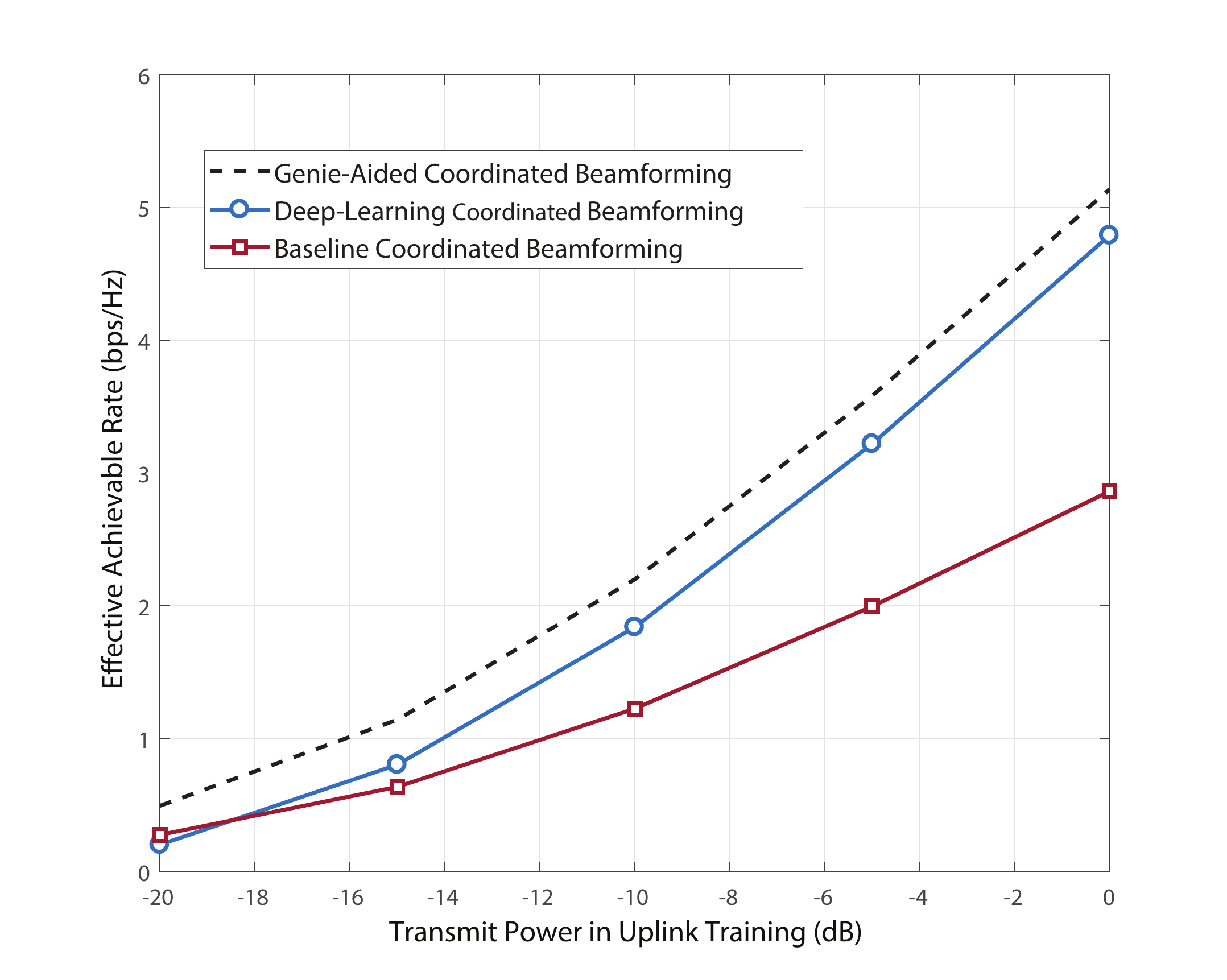}
	}
	\caption{The effective achievable rates of the proposed deep-learning and baseline coordinated beamforming solutions compared with the upper bound, $R^\star$ for different uplink transmit powers. The figure considers the LOS scenario in \sref{subsec:Setup} where 4 BSs are serving a user moving with speed $30$ mph. }
	\label{fig:SNR}
\end{figure}

%---------------------------------------
\textbf{Impact of Uplink Transmit Power and Omni Training Pattern:} 
An important aspect of the proposed deep learning coordinated beamforming solution is the use of only omni (or quasi-omni) beam patterns at the BSs during the uplink training. This raises, though, questions on whether the received signals with omni reception,  $r_{k,n}^{\text{omni}} \forall k$, which are the inputs to the neural network model, have sufficient SNR for the system operation, and whether the MS will need to use very high uplink transmit power to ensure enough receive SNR at the BSs. To answer these questions, we plot the effective achievable rates of the proposed deep learning solution, baseline solution, and optimal bound in \figref{fig:SNR} versus the uplink transmit power. We also assume that the downlink transmit power during data transmission by every BS equals the uplink transmit power from the MS. The rest of the communication system and machine learning parameters are similar to the setup in \figref{fig:Speed}. As shown in \figref{fig:SNR}, for low values of uplink transmit power, the performance of the deep-learning strategy is worse than the baseline solutions, as the SNR of the omni-received sequences is low and the learning model is not able to learn and predict the right beamforming vectors. \textbf{For reasonable uplink transmit power, though, $-10$dBm to $0$dBm, the deep-learning coordinated beamforming achieves good gain over the baseline solution. This means that the receive SNR with omni patterns during uplink training is sufficient to draw a defining RF signature of the environment and efficiently train the neural network model}.

It is important to note here that the main reason why we need to use beamforming during mmWave beam training or channel estimation is to estimate the directional information at every BS, such as the angles of arrival and departure, which we do not need in the proposed deep learning coordinated beamforming system that relies on predicting this information via deep learning using the signals captured at multiple distributed BSs.

%%%%%%%%%%%%%%%%%%%%%%%%%%%%%%%%%%%%%%%%%
\subsection{Impact of Machine Learning Parameters} \label{subsec:ML_impact}
The primary objective of this paper is to motivate leveraging machine learning tools in highly-mobile mmWave communication systems. Optimizing the machine learning model itself, though, is out of the scope of this paper, and is worthy for independent publications. In this section, we briefly highlight the impact of some machine learning parameters, such as the input/output normalization and the neural network architecture, on the system performance.  

\begin{figure}[t]
	\centerline{
		\includegraphics[width=.75\columnwidth,height=260pt]{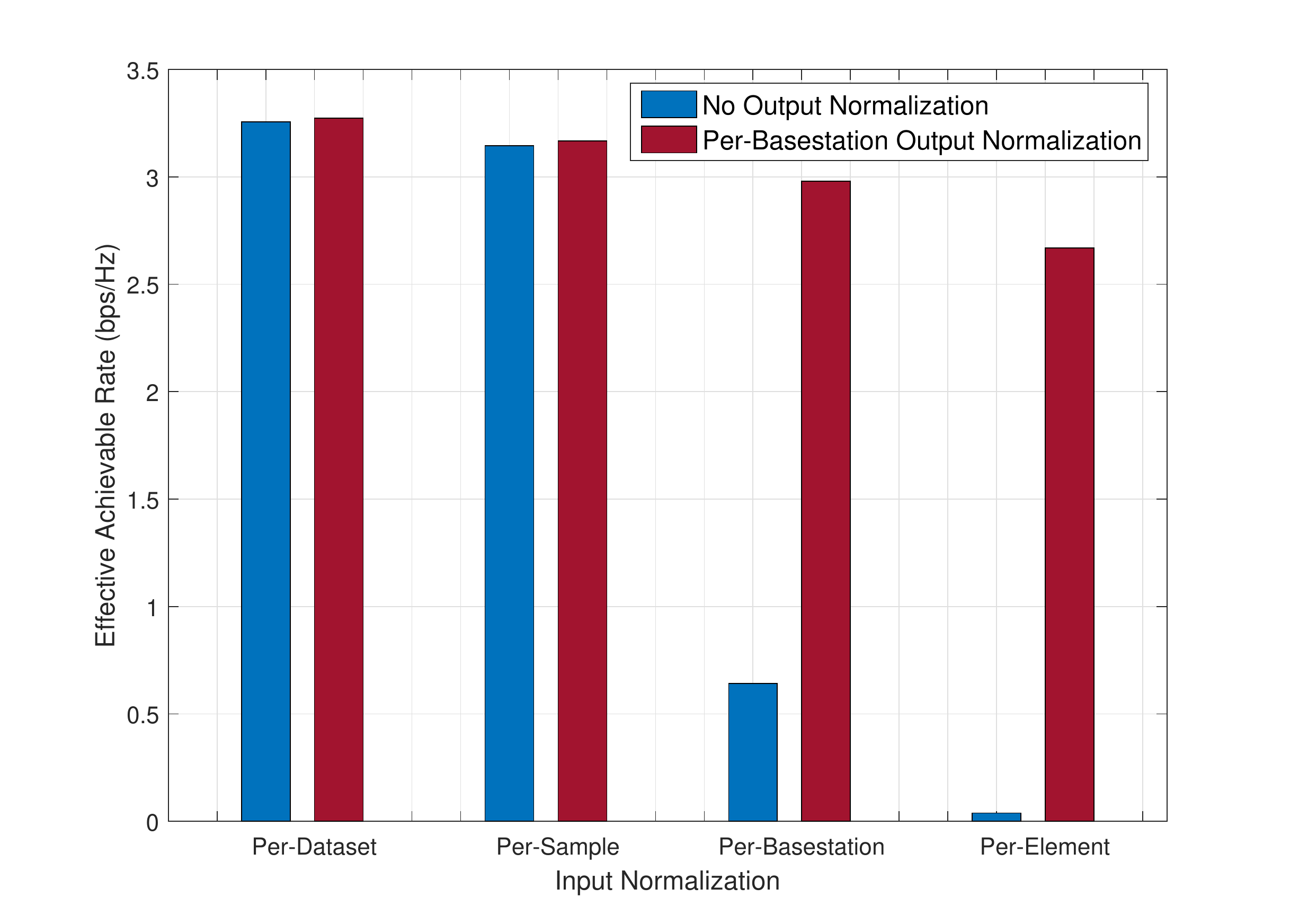}
	}
	\caption{The effective achievable rates for different input and output normalization strategies. This figure considers the NLOS scenario, described in \sref{subsec:Setup} with a deep learning model trained with a 20k samples dataset. The figure shows that per-dataset input normalization in the presence of per-BS output normalization achieves the higher effective rates compared to the other candidate strategies.}
	\label{fig:Norm}
\end{figure}

\textbf{Impact of Input and Output Normalization:}
%---------------------------------------
The proper normalization of the inputs and outputs of the neural network allows realizing efficient machine learning models with high learning rates, robustness against weight initialization biases, among other system gains. In \figref{fig:Norm}, we plot the effective achievable rates for different input normalization strategies, namely per-dataset, per-sample, per-basestation, and per-element normalization, which are explained in detail in \sref{subsec:ML_model}. This figure considers the NLOS scenario, described in \sref{subsec:Setup}, with BSs employing 16$\times$8 UPAs and with a deep-learning model trained using a 20k-samples dataset. As shown in \figref{fig:Norm}, the per-dataset normalization achieves the highest effective achievable rate among the four candidate strategies. To understand the intuition behind this performance, it is important to note that the correlation among the received signals at the different subcarriers of each BS may carry useful information, such as the distance between the user and the BS. Similarly, the correlation between the received signals of the same user at the 4 BSs and the correlation between the received signals at different user locations may carry logical information that helps the neural network model in learning the mapping between the multi-path signatures and the beamforming beams. \textbf{The per-dataset normalization is the only strategy, among the 4 candidates, that preserves all theses kinds of correlation. Therefore, it allows the machine learning model to leverage all the information carried by the training dataset.} 

%---------------------------------------

In \figref{fig:Norm}, we also plot the effective rates with and without per-BS output normalization. The normalization strategy is explained in \sref{subsec:ML_model}.  \figref{fig:Norm} shows that normalizing the outputs of the training dataset is required to achieve good data rates. To justify this performance, we first emphasize that these results consider the NLOS scenario in \figref{fig:setup_NLOS}. In this  scenario, some achievable rates, $R_n^{(p)} \forall p$ (the outputs of the machine learning model), correspond to NLOS links while others are results of LOS links. The challenge here is that the achievable rates corresponding to NLOS links have much smaller values compared to those of LOS links. Without output normalization, the training of the neural network weights will be dominated by the LOS-related outputs which have large differences between its candidate beams (the output bins). These weights will not be sensitive to the relatively small differences between the rates of the NLOS-related outputs. In other words, the machine learning model will only learn how to beamform to the users with LOS links. This draws insights into the importance of normalizing the outputs of the neural network training dataset.

\textbf{Impact of Network Architecture:}
%---------------------------------------
%
In the simulation results of this paper, we adopt the fully-connected neural network architecture in \figref{fig:ML_model}. For the sake of motivating the future research into optimizing the machine learning model, we compare the effective achievable rates of the fully-connected architecture and another architecture based on convolutional neural networks (CNN) in \figref{fig:Arch}. This figure adopts the LOS scenario with BSs employing $32 \times 8$ UPAs and steering codebooks with oversampling factor of 2 in the azimuth direction. The fully-connected architecture consists of $4$ layers with 512 nodes per layer. For the CNN-based architecture, it first applies four 2D $32 \times 2$ convolutional filters on two input channels representing the real and imaginary of the omni-received sequences $r_{k,n}^{\text{omni}} \forall k$. A max-pooling layer is then added and followed by three fully-connected layers with 512 nodes per layer. This results in a total of $\sim$754k parameters in the CNN-based architecture compared to $\sim$ 1048k parameters in the fully-connected architecture. Despite its lower complexity compared to the fully-connected architecture, the CNN architecture achieves almost the same effective spectral efficiency of the fully-connected model, as shown in \figref{fig:Arch}. One intuition for this efficient performance comes from the CNN dependence on extracting local information using  small-sized filters. In our model, these filters may capture the correlation between the adjacent samples in the OFDM sequence, which helps extracting valuable information with lower complexity compared to the brute-force approach in the fully-connected model.  This highlights the potential of exploring new neural network architectures for integrated learning/communication systems. 

%---------------------------------------
\begin{figure}[t]
	\centerline{
		\includegraphics[width=.75\columnwidth,height=260pt]{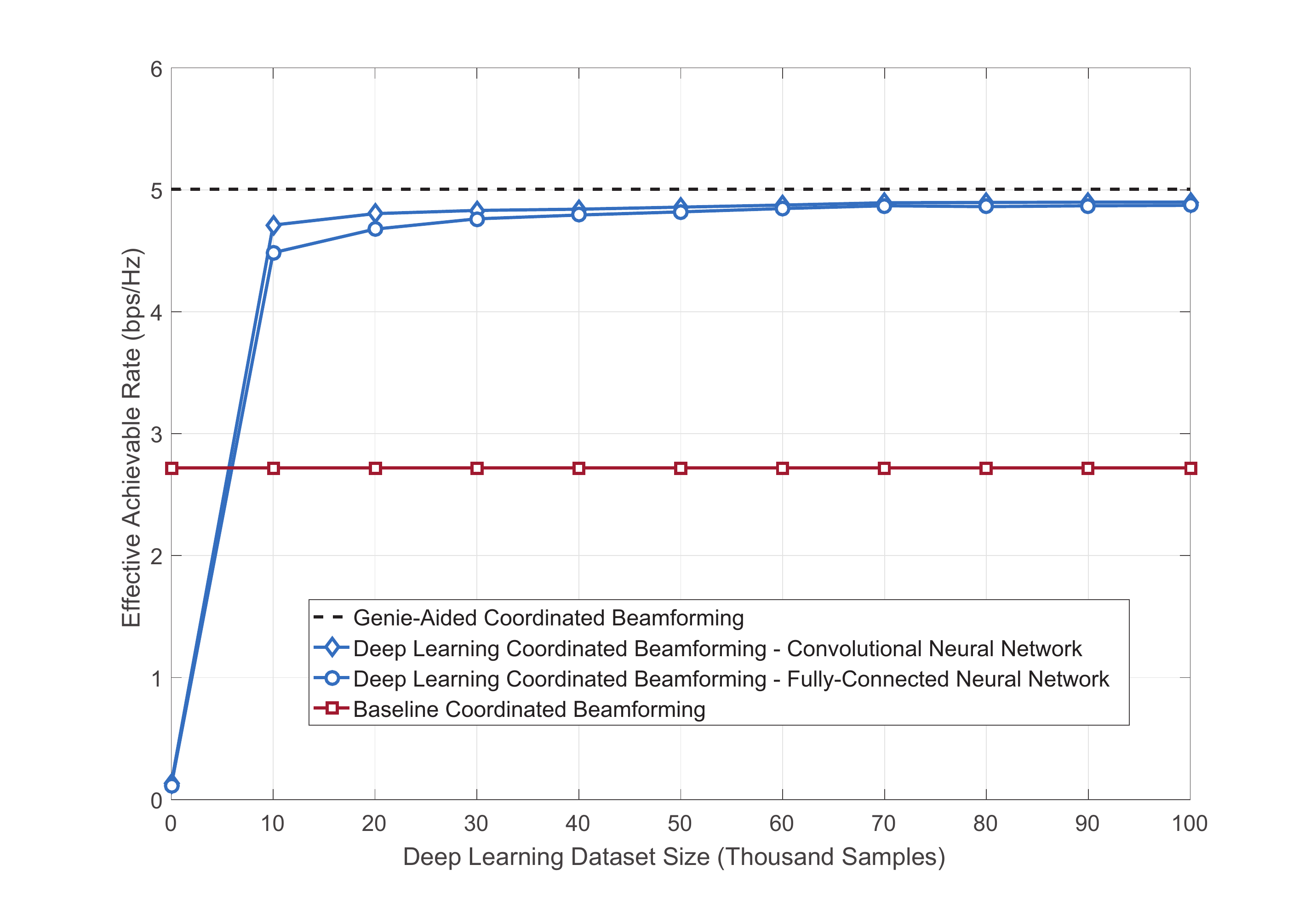}
	}
	\caption{The figure compares the effective achievable rates of the adopted fully-connected neural network architecture in \figref{fig:ML_model} and another architecture based on CNN. The results show the two architectures achieve almost the same effective data rates despite the potential complexity reduction in the CNN model.}
	\label{fig:Arch}
\end{figure}

%%%%%%%%%%%%%%%%%%%%%%%%%%%%%%%%%%%%%%%%%
\subsection{System Adaptability and Robustness} \label{subsec:Adapt}

\begin{figure}[t]
	\centerline{
		\includegraphics[width=.75\columnwidth]{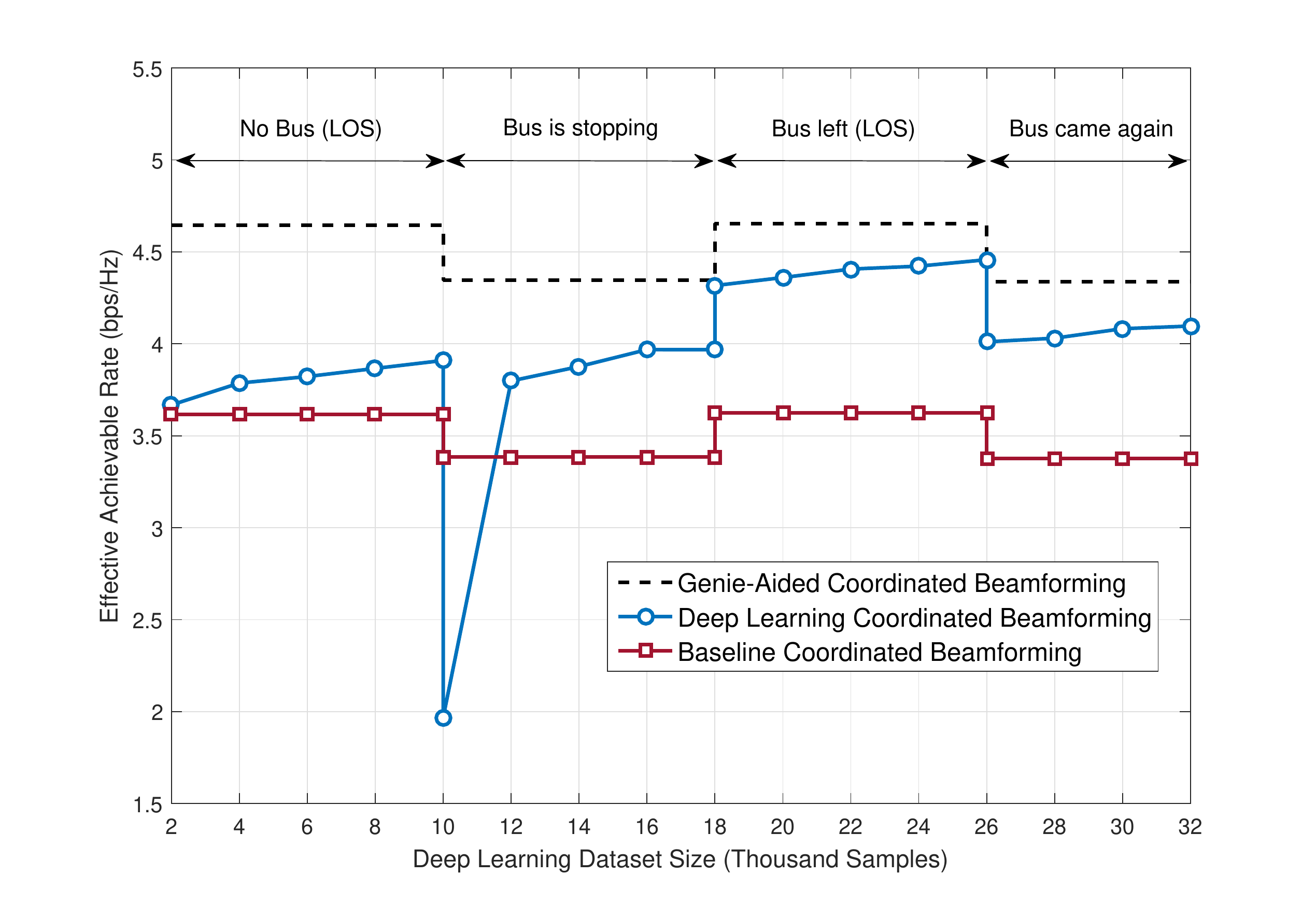}
	}
	\caption{The figure plots the effective achievable rates of the developed deep-learning and baseline coordinated beamforming as well as the upper bound, $R^\star$, for a setup where the environment is changing between LOS and NLOS scenarios. The figure illustrates that the deep learning model generalizes its learning over time to perform well at both LOS and NLOS scenarios.}
	\label{fig:Adapt}
\end{figure}

One main advantage of integrating machine learning in wireless communication is realizing robust systems that adapt efficiently to the highly-mobile aspects of the environment. 
To examine this gain, we plot the effective achievable rates in \figref{fig:Adapt} for an important setup where the environment changes multiple times as follows. 
\begin{itemize}
	\item First, when the system started working, at dataset size equals 0 samples, the LOS scenario in \figref{fig:setup_LOS} was considered where 4 BSs is serving a car  moving alone in the street. The BSs employ $32 \times 8$ UPAs and using beamsteering codebooks with oversampling factor of 2 in only the azimuth direction.
	\item After some time, which is spent to build a dataset of size $10$k samples, a large bus appeared suddenly and stopped in front of BS 3, as depicted in \figref{fig:setup_LOS}. Since the deep-learning model was trained only for the LOS scenario before the bus arrives, the effective achievable rate of the deep-learning coordinated beamforming solution degraded significantly at the first moment of the bus arrival. This is clear in the effective rate transition at dataset size $10$k samples in \figref{fig:Adapt}. Assuming that the bus parked in front of BS 3 for some time, the deep learning model started learning this new NLOS scenario. In other words, the neural network weights that were initially adjusted for the LOS dataset are now being refined again based on the new  NLOS training samples. 
	\item After more time, which is spent to build an overall dataset of size $18$k samples, the bus left. Interestingly, the performance of the proposed deep learning solution now did not degrade again, but rather did as well as the first stage (before the bus arrives). \textbf{This is very important as it shows that the deep learning model has generalized its learning to both the LOS and NLOS scenarios}, which is also confirmed by the performance of the deep-learning solution after the bus arrives again at the dataset size $26$k samples.   
\end{itemize}

\noindent The results in \figref{fig:Adapt} show that the coordinated beamforming system became more robust over time, and is able to adapt and perform well at both the LOS and NLOS scenarios. \textbf{More generally, this means that when we first deploy the deep-learning coordinated beamforming system in a new environment, it will experience many new scenarios, such as cars and pedestrian blocking the signals, trees growing, etc., for which the system was not trained. After some time, the model will generalize its learning to cover all these scenarios and develop into a robust and adaptable system.}

%%%%%%%%%%%%%%%%%%%%%%%%%%%%%%%%%%%%%%%%%
\subsection{Does the System Require Phase Synchronization to Learn?} \label{subsec:sync}

\begin{figure}[t]
	\centerline{
		\includegraphics[width=.75\columnwidth,height=260pt]{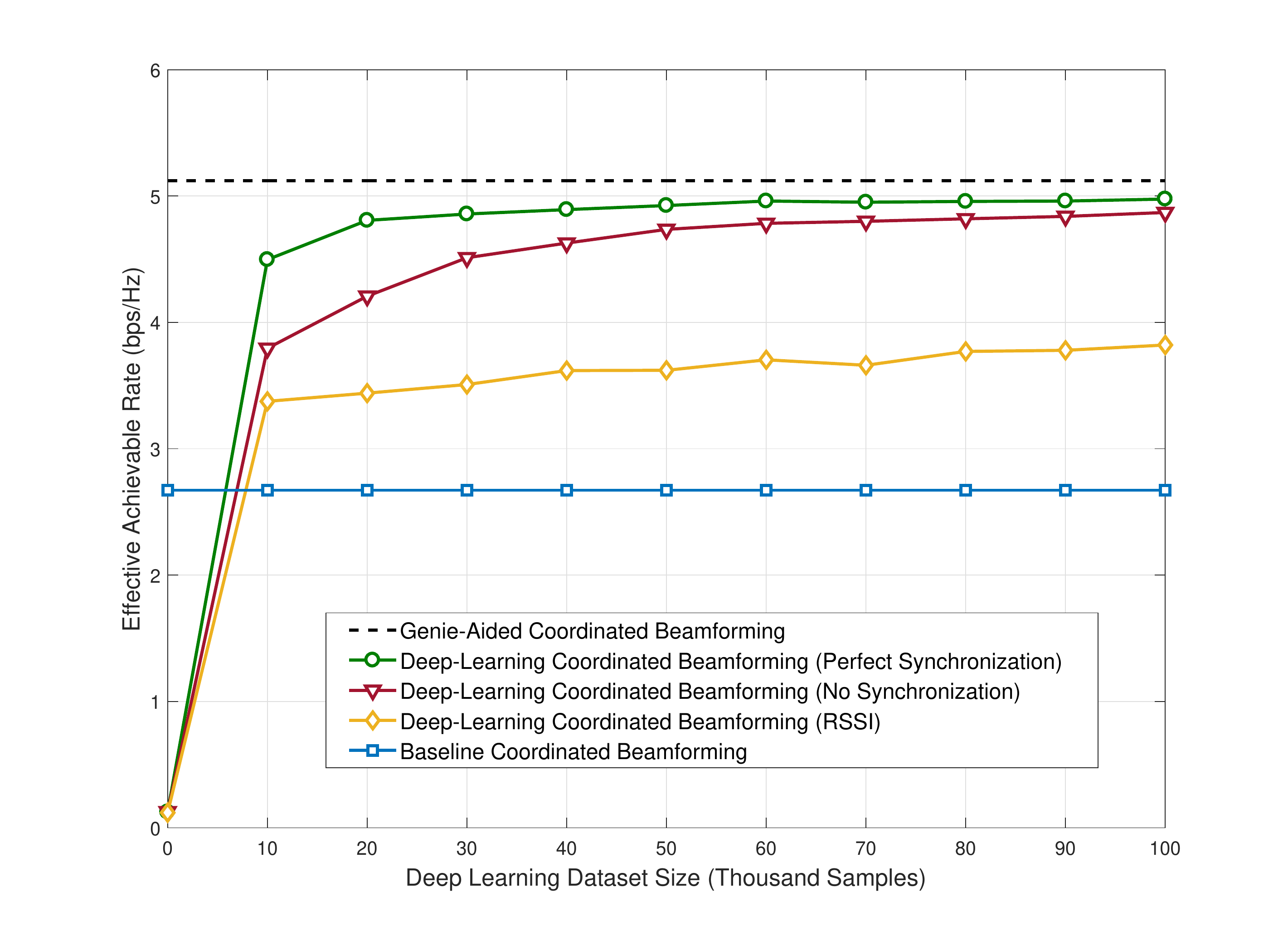}
	}
	\caption{The effective achievable rates of the proposed deep-learning coordinated beamforming solution are plotted for different phase synchronization assumptions. These rates are also compared with the baseline solution and the upper bound, $R^\star$. The figure shows that the deep-learning coordinated beamforming achieves good gain over the baseline solution even when only received signal strength indicators (no phase information) are used as inputs to the machine learning model.}
	\label{fig:Sync}
\end{figure}

The machine learning model, in the proposed deep-learning beamforming solution, relies on the signals received \textit{jointly} at multiple BSs. Therefore, the phase of these signals may intuitively carry useful information that helps the model in learning how to predict the beamforming for each multi-path signature. Maintaining this phase information, though, is difficult in practice as it requires perfect synchronization of the terminal BSs oscillators. In this section, we are interested in evaluating the performance of the proposed deep-learning coordinated beamforming solution in a setting where we relax the phase synchronization requirements. 

In \figref{fig:Sync}, we consider the LOS scenario in \sref{subsec:Setup}, and plot the effective achievable rates of the proposed deep learning coordinated beamforming solution under three different assumptions on the phase synchronization: (i) perfect phase synchronization where the clocks of 4 BSs are perfectly synchronized, (ii) no synchronization, where uniform random phase $\delta_n \in \left[0, 2 \pi\right]$ is added to the omni received signal at every BS $n$, and (iii) received signal strength indicators (RSSI), where only the amplitude of the omni received sequence, $|r_{k,n}^{\text{omni}} |,  \forall k,n$, is fed to the neural network model. As shown in \figref{fig:Sync}, the performance of the deep-learning coordinated beamforming with no phase synchronization approaches that with perfect phase synchronization as more time is spent in training the neural network (or equivalently large datasets are adopted). \textbf{This result is very useful for practical implementations as it means that the phase synchronization may not be needed to learn coordinated beamforming if large enough datasets are adopted}. \figref{fig:Sync} also illustrates that relying only on RSSI in deep-learning coordinated beamforming, which does not require any phase information,  still achieves a reasonable gain over the baseline coordinated beamforming solution. 

Finally, it is worth mentioning that while \figref{fig:Sync} shows that the machine learning model can learn well with no phase synchronization, both the baseline and the deep-learning coordinated beamforming solutions still need this synchronization in the downlink data transmission phase, as the signals from the 4 BSs need to add coherently at the mobile user antenna. This requirement though can be relaxed if the user is served with only one BS at a time. This way, the 4 BS coordinate the learning but only one of them beamform to the user at any given time. Clearly, these different approaches for coordinated beamforming have a trade-off between implementation complexity and system performance (data rate, reliability, etc.). Investigating this trade-off for practical systems is an interesting future research direction.

%%%%%%%%%%%%%%%%%%%%%%%%%%%%%%%%%%%%%%%%%%%%%%%%%%%%%%%%%%%%
\section{Conclusion} \label{sec:Conclusion}
%%%%%%%%%%%%%%%%%%%%%%%%%%%%%%%%%%%%%%%%%%%%%%%%%%%%%%%%%%%%
In this paper, we developed an integrated machine learning and coordinated beamforming strategy that enables highly-mobile applications in large antenna array mmWave systems. The key idea of the developed strategy is to leverage a deep learning model that learns the mapping from omni-received uplink pilots and the beam training result. This is motivated by the intuition that the signal received at multiple distributed BSs renders an RF defining signature for the user location and its interaction with the surrounding environment.  The proposed solution requires negligible training overhead and performs almost as good as the genie-aided solution that perfectly knows the optimal beamforming vectors. Further, thanks to integrating deep learning with the coordinated transmission from multiple BSs, the developed solution ensures reliable coverage and low latency, resulting in a comprehensive framework to enable highly-mobile mmWave applications. Extensive simulations, based on accurate ray-tracing, were performed to evaluate the proposed solution in various LOS and NLOS environment.  These results indicated that the proposed solutions attains high data rate gains compared to coordinated beamforming strategies that do not leverage machine learning, especially in high-mobility large-array scenarios. The results also illustrated that with sufficient learning time, the deep learning model efficiently adapts to changing environment, yielding a robust beamforming system. From a practical perspective, the results illustrated that phase synchronization among the coordinated BSs is not necessary for learning how to accurately predict the beamforming vectors. The results in this paper encourage several future research directions such as the extension to multi-user systems, the investigation of time-varying scenarios, and the development of more sophisticated machine learning models for mmWave beamforming.

\bibliographystyle{IEEEtran}
% Generated by IEEEtran.bst, version: 1.14 (2015/08/26)

\end{document}